\documentclass[12pt]{article}
\pdfoutput=1

\newlength{\abstractwidth}
\usepackage{amscd,amsmath,amssymb,amsfonts,xspace,mathrsfs,amsthm,dsfont,bbm}
\usepackage{color}

\usepackage{bbold}
\usepackage{latexsym}
\usepackage{graphicx}
\usepackage{dsfont}
\usepackage{color}
\usepackage{longtable}
\usepackage{hyperref}

  \newtheorem{proposition}{Proposition}
  
  \newtheorem{conj}{Conjecture}

\numberwithin{equation}{section}

\def\bea{\begin{eqnarray}}
\def\eea{\end{eqnarray}}
\def\be{\begin{equation}}
\def\ee{\end{equation}}
\def\ba{\begin{align}}
\def\ea{\end{align}}
\def\bse{\begin{subequations}}
\def\ese{\end{subequations}}

\newcommand{\nn}{\nonumber}

\newcommand{\Li}{{\rm Li}}

\def\Sym{\,{\rm Sym}\, }

\def\Im{\,{\rm Im}\,}
\def\Re{\,{\rm Re}\,}

\def\({\left(}
\def\){\right)}
\def\[{\left[}
\def\]{\right]}
\def\<{\left\langle}
\def\>{\right\rangle}
\def\hf{{1\over 2}}

\newcommand{\p}{\partial}
\newcommand{\unit}{{\mathbbm{1}}}

\renewcommand{\d}{\mathrm{d}}
\newcommand{\de}{\mathrm{d}}

\newcommand{\I}{\mathrm{i}}

\newcommand{\eps}{\epsilon}

\newcommand{\vth}{\vartheta}
\newcommand{\hk}{{hyperk\"ahler}\xspace}
\newcommand{\qk}{{quaternion-K\"ahler}\xspace}

\newcommand{\cD}{\mathcal{D}}

\newcommand{\cS}{\mathcal{S}}

\newcommand{\cM}{\mathcal{M}}
\newcommand{\cW}{\mathcal{W}}
\newcommand{\cN}{\mathcal{N}}

\newcommand{\cX}{\mathcal{X}}

\newcommand{\cT}{\mathcal{T}}
\newcommand{\cJ}{\mathcal{J}}
\newcommand{\cZ}{\mathcal{Z}}

\newcommand{\cO}{\mathcal{O}}

\newcommand{\cU}{\mathcal{U}}

\newcommand{\IT}{\mathds{T}}
\newcommand{\IR}{\mathds{R}}
\newcommand{\IC}{\mathds{C}}
\newcommand{\IZ}{\mathds{Z}}

\newcommand{\IP}{\mathds{P}}

\newcommand{\sgn}{\mbox{sgn}}

\def\tcJ{\tilde\cJ}

\def\cl0{\tilde c_0}

\def\bOm{\bar\Omega}

\def\ba{\bar a}

\def\bZ{\bar Z}

\def\CY{\mathfrak{Y}}

\DeclareMathOperator{\ch}{ch}

\def\cXsf{\cX^{\rm sf}}
\def\cXcl{\cX^{\rm sf}}

\def\Hcl{\cX^{\rm cl}}

\def\cXr{\cX^{\rm (ref)}}
\def\Kr{K^{\rm (ref)}}
\def\Wr{W^{\rm (ref)}}
\def\Fr{F^{\rm (ref)}}

\def\dm{m}
\def\fl#1{\lfloor #1\rfloor}
\def\IS{\lefteqn{\textstyle\sum}\int}

\begin{document}

\thispagestyle{empty}

\begin{flushright}
2104.10540v2
\end{flushright}
\vskip 0.3in

\begin{center}
{\Large \bf Heavenly metrics, BPS indices and twistors}
\vskip 0.2in

{\large Sergei Alexandrov$^\dagger$ and Boris Pioline$^\ddagger$}

\vskip 0.15in

$^\dagger$ {\it
Laboratoire Charles Coulomb (L2C), Universit\'e de Montpellier,
CNRS, F-34095, Montpellier, France}

\vskip 0.15in

$^\ddagger$ {\it Laboratoire de Physique Th\'eorique et Hautes
Energies (LPTHE), UMR 7589 CNRS-Sorbonne Universit\'e,
Campus Pierre et Marie Curie,
4 place Jussieu, F-75005 Paris, France} \\

\vskip 0.15in

{\tt \small sergey.alexandrov@umontpellier.fr,  pioline@lpthe.jussieu.fr}

\vskip 0.2in

\begin{abstract}
Recently T. Bridgeland
defined a complex hyperk\"ahler metric on the tangent bundle over the space of stability
conditions of a triangulated category, based on a Riemann-Hilbert problem determined by the Donaldson-Thomas invariants.
This metric is encoded in a function $W(z,\theta)$ satisfying a heavenly equation, or
a potential $F(z,\theta)$ satisfying an isomonodromy equation.
After recasting the RH problem into a system of TBA-type equations,
we obtain integral expressions for both $W$ and $F$ in terms of solutions of that system. These expressions are
recognized as conformal limits of the `instanton generating potential' and `contact potential'
appearing in studies of D-instantons and BPS black holes.
By solving the TBA equations iteratively,
we reproduce Joyce's original construction of $F$ as a formal series in the rational DT invariants.
Furthermore, we produce similar solutions to deformed versions of the heavenly and isomonodromy equations
involving a non-commutative star-product.
In the case of a finite uncoupled BPS structure, we rederive the results previously obtained
by Bridgeland and obtain the so-called $\tau$ function for arbitrary values of the fiber coordinates $\theta$,
in terms of a suitable two-variable generalization of Barnes' $G$ function.
\end{abstract}
\end{center}

\newpage
\setcounter{tocdepth}{2}
\tableofcontents

\baselineskip=15pt
\setcounter{page}{1}
\setcounter{equation}{0}
\setcounter{footnote}{0}
\setlength{\parskip}{0.2cm}

\section{Introduction}

In the context of supersymmetric field theories and string vacua with 8 supercharges,
it is well-known that instanton corrections to the metric on certain quaternionic\footnote{Since  \hk (HK) and  \qk (QK)
spaces  are parametrized by solutions of non-linear differential equations often known as heavenly equations,
we refer to such quaternionic metrics as `heavenly', irrespective of their precise nature.
See \cite{Plebanski:1975wn} and \cite{Przanowski:1984qq} for the original descriptions in
the four-dimensional HK and QK cases, and \cite{Adamo:2021bej} for a recent discussion of the higher dimensional case.}
moduli spaces are determined from the BPS indices $\Omega(\gamma,z)$ which count BPS states with
electromagnetic charge $\gamma$ in $3+1$ dimensions. This includes the  metric on the vector-multiplet
moduli space $\cM$ of 4D gauge theories compactified on a circle of radius $R$. The latter is a torus
bundle over the Coulomb branch $\cS$ in $3+1$ dimensions, equipped with a hyperk\"ahler metric which reduces to
the semi-flat metric in the limit $R\to \infty$, but is deformed by $\cO(e^{-R})$
instanton corrections induced by BPS dyons whose worldline winds around the circle \cite{Seiberg:1996nz}.
The full quantum corrected metric is  determined in terms of the BPS indices $\Omega(\gamma,z)$ via
a twistorial construction  due to \cite{Gaiotto:2008cd} (GMN). A key fact is that the HK metric on $\cM$ is smooth
across walls of marginal stability, even though the BPS indices $\Omega(\gamma,z)$ are discontinuous.

Similarly, for type II strings compactified on a $\cX\times S^1$ where $\cX$ is a Calabi-Yau (CY) threefold,
instanton corrections to the metric on the vector-multiplet moduli space $\cM$ are governed
(up to subleading $\cO(e^{-R^2})$ corrections due to Taub-NUT instantons)
by the generalized Donaldson-Thomas (DT) invariants $\Omega(\gamma,z)$  of the relevant triangulated category
of branes on  $\cX$ (the category of coherent sheaves
on $\cX$ in type IIA, or  the Fukaya category of Lagrangian submanifolds in type IIB), which count
BPS black holes in $3+1$ dimensions \cite{Douglas:2000gi}.
The space $\cM$ is a twisted torus bundle over $\IR^+_R\times \cS$, where $\cS$ is now the
moduli space of complexified K\"ahler or complex structures of $\cX$, respectively.
At order $\cO(e^{-R})$, the instanton-corrected \qk metric on $\cM$
is determined via an analogue of the GMN construction for QK spaces \cite{Alexandrov:2008gh,Alexandrov:2009zh}
(see \cite{Alexandrov:2011ac} for a precise correspondence between the two twistorial constructions,
and \cite{Alexandrov:2011va,Alexandrov:2013yva} for reviews and further references).
By T-duality, the same BPS indices also determine instanton corrections to
the hypermultiplet moduli space in type II strings compactified on
$\cX$ down to 4 dimensions, or in M-theory compactified on $\cX$ down to 5 dimensions.

Recently, Bridgeland proposed to associate to any triangulated category $\cD$
of CY3 type a complex \hk cone metric on the total space $\cM$ of the holomorphic tangent bundle to the space
of stability conditions $\cS={\rm Stab}(\cD)$ \cite{Bridgeland:2019fbi,Bridgeland:2020zjh}.
This metric is determined by the DT invariants $\Omega(\gamma,z)$, where $z$
parametrizes the space of stability conditions,
and is a close analogue of the metric obtained from the GMN construction in the so-called conformal limit \cite{Gaiotto:2014bza}.
However, a key difference is that the metric (being complex) has split signature, both along the base and along the fiber.
Moreover, the base $\cS$ of the fibration is the full space of stability conditions
${\rm Stab}(\cD)$  in the sense of \cite{MR2373143}, rather than the complex Lagrangian submanifold
parametrizing the moduli space of complexified K\"ahler or complex structures of a CY threefold $\cX$,
in the case where $\cD$ is one of the afore-mentioned categories of branes.
For these reasons, the physical significance of Bridgeland's construction remains
unclear.\footnote{In the context of class $S$ field theories, the full space of stability conditions
${\rm Stab}(\cD)$ \cite{bridgeland2015quadratic} can be interpreted as the total space of the Coulomb branch fibration over
the conformal manifold, but the significance of the complex HK metric on $\cM$ is yet to
be understood.} Nevertheless, given its mathematical relevance and close
connections with the string and gauge theory constructions,
we find it worthwhile to examine it with the same tools as previous,
physically motivated heavenly moduli spaces. As  we shall see,
its basic ingredients turn out to have direct physical counterparts which have been defined on those moduli spaces.

As explained in \cite{Bridgeland:2020zjh,Dunajski:2020qhh}, the main step in constructing
the complex \hk metric on $\cM=T\cS$ is to produce a solution $W(z,\theta)$
of the partial differential equations\footnote{The factor $1/(2\pi)^2$ appearing in this and a few other equations is due to
our normalization of the coordinates $z^a$ and $\theta^a$ which coincides with normalization used in \cite{Bridgeland:2017vbr},
but differs by the factor $2\pi\I$ from the one in \cite{Bridgeland:2019fbi,Bridgeland:2020zjh}.}
\be
\frac{\p^2 W}{\p z^a\p\theta^b}-\frac{\p^2 W}{\p z^b\p \theta^a}
=\frac{1}{(2\pi)^2} \sum_{c,d=1}^{\dm}
\omega^{cd}\, \frac{\p^2 W}{\p \theta^a\p\theta^c}\, \frac{\p^2 W}{\p \theta^b\p\theta^d}\, ,
\qquad \forall a,b=1,\dots, \dm,
\label{heaveq}
\ee
where $z^a$ are complex coordinates on $\cS$, $\theta^a$ are
the corresponding complex coordinates on the fiber, and $\omega^{ab}=\langle\gamma^a,\gamma^b\rangle$
is a constant, invertible antisymmetric matrix determined by the skew-symmetric pairing
on the charge lattice $\Gamma=\IZ \gamma^1+\dots + \IZ\gamma^m$ (which implies that $m$ is even).
For $m=2$, \eqref{heaveq} is known as the `second heavenly equation' governing generic
HK manifolds in real dimension 4 \cite{Plebanski:1975wn}.
The solution $W(z,\theta)$
constructed by Bridgeland is  homogeneous of degree $-1$ with respect to the coordinates $z^a$,
leading to a cone metric on  $\cM=T\cS$.
As shown in \cite{Bridgeland:2020zjh}, its derivative
\be
F=\sum_{a=1}^{\dm} z^a\p_{\theta^a} W\ ,
\label{FW}
\ee
which is homogeneous of degree $0$ with respect to the coordinates $z^a$,
turns out to be identical to the holomorphic function postulated by Joyce in his prescient work \cite{joyce2007holomorphic}.
We shall refer to $W$ and $F$ as the Pleba\'nski and Joyce potentials, respectively.\footnote{In
\cite{Bridgeland:2019fbi,Bridgeland:2020zjh} $W$ has been called the Joyce function, but we prefer to reserve this attribution
for $F$.}

The equation \eqref{heaveq}
arises as the integrability condition for the vector fields on the twistor space $\cZ=
\IP^1_t\times \cM$,
\be
u_a = \frac{\p}{\p z^a}+\frac{1}{t}\, \frac{\p}{\p\theta^a}
-\frac{1}{(2\pi)^2} \sum_{b,c=1}^{\dm}\omega^{bc}\frac{\p^2 W}{\p \theta^a\p\theta^b}\,\frac{\p}{\p\theta^c}
\label{actonX}
\ee
whose flat sections are the holomorphic Fourier modes\footnote{Here $\sigma_{\gamma}$
is a quadratic refinement of the skew symmetric pairing on the
charge lattice $\Gamma$, a technical device which allows to trivialize the twisted group law
$\cX_\gamma\,\cX_{\gamma'} =(-1)^{\langle\gamma,\gamma'\rangle} \cX_{\gamma+\gamma'}$.}
$\cX_\gamma(t)=\sigma_{\gamma} e^{2\pi\I \langle\gamma, \xi(t)\rangle}$,
where the complex Darboux coordinates $\xi^a(t)=\langle\gamma^a,\xi(t)\rangle$ parametrize the twistor lines on  $\cZ$.

In this work, we shall provide general formulae for the Pleba\'nski and Joyce potentials,
and relate them to quantities which were introduced previously in the physical context of quaternionic moduli spaces.
Our formulae are expressed in terms of solutions to
the TBA-like system of integral equations for the holomorphic Darboux coordinates  $\cX_\gamma$
on $\cZ$,
\be
\cX_\gamma(t)=
\sigma_{\gamma} \,e^{2\pi\I(\theta_{\gamma}-Z_\gamma/t)}
\exp\[\frac{1}{2\pi\I}\sum_{\gamma'\in\Gamma}
\Omega(\gamma')\langle\gamma,\gamma'\rangle\int_{\ell_{\gamma'}}\frac{\de t'}{t'}\,
\frac{t}{t-t'}\,\log\(1-\cX_{\gamma'}(t')\)\]
\label{TBAeq}
\ee
where $\ell_\gamma$ are the BPS rays on $\IP^1_t$.
This system, first proposed in \cite{Gaiotto:2014bza,Filippini:2014sza}, arises  in the conformal
limit of the integral equations in \cite{Gaiotto:2008cd}. It is
equivalent to the Riemann-Hilbert problem introduced in \cite{Bridgeland:2016nqw}
provided the sum over $\gamma'\in\Gamma$ converges. Throughout this work, we always assume
that such sums over charge vectors are  convergent (possibly not in absolute sense, but for
an appropriate prescription). This is the case, for example, when
only a finite number of indices $\Omega(\gamma)$ are non-vanishing. Examples where
an infinite number of indices contribute will be discussed in our subsequent paper \cite{Alexandrov:2021prq}.

Our main result shown in \S\ref{sec_int} is that the integral
\be
\begin{split}
W=&\, \frac{1}{2\pi\I} \sum_{\gamma\in\Gamma} \Omega(\gamma)
\int_{\ell_{\gamma}}\frac{\de t}{t^2}\,\Li_2\(\cX_{\gamma}(t)\)
\\
&\,
+\frac{1}{2(2\pi\I)^2} \sum_{\gamma, \gamma'\in \Gamma}
\Omega(\gamma)\, \Omega(\gamma')
\int_{\ell_{\gamma}}\frac{\de t}{t}\
\int_{\ell_{\gamma'}}\frac{\de t'}{t'}\,
\frac{\langle\gamma,\gamma'\rangle}{t-t'}\,\log\(1-\cX_{\gamma}(t)\)\log\(1- \cX_{\gamma'}(t')\)
\end{split}
\label{Wfull}
\ee
provides a solution to the heavenly equation \eqref{heaveq}, so that the holomorphic Fourier modes  $\cX_\gamma$
are annihilated by the vector fields \eqref{actonX}. Moreover, we also show that the Joyce potential
is given by an even simpler formula,
\be
\label{Ffull}
F=-\sum_{\gamma\in\Gamma}\Omega(\gamma)\, Z_{\gamma}
\int_{\ell_{\gamma}}\frac{\de t}{t^2}\,\log\(1-\cX_{\gamma}(t)\) .
\ee
By construction, $F$ and $W$ are  smooth across walls of marginal stability.

Remarkably, functions similar to $F$ and $W$ have already appeared previously.
In particular, comparing \eqref{Ffull} with the contact potential introduced in \cite{Alexandrov:2008nk} and computed in
the case relevant to our setup in \cite[Eq.(3.41)]{Alexandrov:2009zh}, we see
that the Joyce potential \eqref{Ffull} is essentially the conformal limit of the (non-perturbative part of the) contact potential.
In the context of string theory vacua, the latter is identified with the four-dimensional dilaton,
whereas in the gauge theory context, it was argued to be identical with
the Witten index on $\IR^3$ \cite{Alexandrov:2014wca}.
Mathematically, the contact potential is closely related to the K\"ahler potential on the twistor space,
or equivalently to the \hk potential on the Swann bundle  \cite{Alexandrov:2008nk}.
Similarly, the Pleba\'nski potential \eqref{Wfull} can be viewed as the conformal
limit of the instanton generating potential introduced in the study of D3-instanton corrections \cite[Eq.(3.22)]{Alexandrov:2018lgp},
where it played a crucial role in determining the modular properties of generating functions of DT invariants.
Furthermore, both potentials have an interpretation in the context of integrable systems described by TBA equations:
comparing with \cite[Eqs.(3.5),(3.9)]{Alexandrov:2010pp} it is easy to see that $F$ coincides with
the conformal limit of the free energy, while $W$ is the conformal limit of the critical value of the
Yang-Yang functional.

Upon solving the TBA equation \eqref{TBAeq} iteratively and plugging into \eqref{Wfull} and \eqref{Ffull},
we obtain expansions for $F$ and $W$ in powers of  the rational DT invariants  $\bOm(\gamma_i,z)$.
The coefficients in these expansions involve certain iterated integrals which are reminiscent of the Goncharov polylogarithms
found in a similar context in \cite{bridgeland2013stokes}, but have an arguably simpler form.
Subject to a purely combinatorial conjecture which has passed many tests,
we show that these coefficients satisfy a set of conditions determined by Joyce \cite{joyce2007holomorphic},
and given the uniqueness of the solution to these conditions, we conclude that they
coincide with the ones found in \cite{bridgeland2013stokes}.
It should be pointed out that these expansions are formal since, unlike in the non-conformal case \cite{Gaiotto:2008cd},
there is no exponential damping factor at large charge.

We also analyze the $\tau$ function introduced in \cite{Bridgeland:2016nqw,Bridgeland:2017vbr}
as a kind of generating function for Darboux coordinates in the special case of uncoupled BPS structures,
such that $\langle \gamma,\gamma'\rangle=0$ whenever $\Omega(\gamma,z)$ and $\Omega(\gamma',z)$ are both non-zero.
In this case, we can integrate the defining differential equation explicitly and obtain an integral representation
for $\log\tau$ which extends the previous results to non-vanishing $\theta^a$.
If the BPS structure is more general, the integrability condition for the differential equation
is not satisfied and thus the $\tau$ function does not exist, although we speculate
that this could be cured by appropriately modifying the defining equation.

In the case of uncoupled BPS structures, wall-crossing phenomena are absent and the conformal
TBA equations can be solved explicitly  in terms of the Euler Gamma function \cite{Gaiotto:2014bza,Barbieri:2018swu}.
As a result, both $F$ and $W$ evaluate to sums over charges $\gamma$ such that $\Omega(\gamma,z) \neq 0$.
In the case where $\Omega(\gamma)$ is supported on a finite number of charge vectors,
this allows us to recover the results of \cite{Bridgeland:2019fbi}.
The $\theta^a$-dependent $\tau$ function is captured in this case by a two-parameter
generalization $\Upsilon(z,\eta)$ \eqref{defUpsilon} of the Barnes function, which reduces
to the function
introduced in \cite{Bridgeland:2016nqw} when $\eta=0$.
As a step in deriving the asymptotic expansion of this function at large $z$, we establish a new integral
formula \eqref{BinetG1gen} for the Barnes function which is analogous to Hermite's generalization
of the first Binet formula for the Gamma function.

If the BPS spectrum is infinite, the situation is much more complicated because
the sum over charges in the formula \eqref{taufun-res2} for the $\tau$ function is {\it not} absolutely convergent
and  therefore ambiguous. We address this case and compare with the results
of \cite{Bridgeland:2017vbr,Bridgeland:2019fbi} for the resolved conifold in our subsequent
work \cite{Alexandrov:2021prq}.

Finally, based on our previous work \cite{Alexandrov:2019rth},
we provide a general solution for a deformation of \eqref{heaveq}, where the Poisson bracket on the r.h.s.
is replaced by a Moyal commutator, in terms of solutions of a suitably deformed TBA-like system of equations
governed by refined DT invariants $\Omega(\gamma,y,z)$. The smoothness of the solution across walls of
marginal stability is again ensured by construction.
On the other hand, the smoothness in the unrefined limit $y\to 1$ is not manifest.
Nevertheless, relying on a technical conjecture which has been extensively checked on a computer,
we show that our solution does reduce to \eqref{Wfull} in this limit.
The mathematical and physical significance of these deformations are not yet fully appreciated,
and it would be interesting to relate our results to other approaches such as \cite{Cecotti:2014wea,Barbieri:2019yya}.

The outline of this article is as follows.
In the next section we review the notion of BPS structure, the twistorial formulation of the associated HK geometry,
and the definitions and some properties of the Joyce and Pleba\'nski potentials.
In \S\ref{sec_joyce} we prove the integral representations for these two potentials presented above,
derive their formal expansions in powers of DT invariants, and analyze the $\tau$ function.
In \S\ref{sec-local} we make the previous results explicit in the uncoupled case,
and in \S\ref{sec_def} we provide a refined version of our construction.
Several appendices contain proofs and details omitted in the main text.
In particular, Appendix \ref{ap-proof} collects most of our proofs.
In Appendix \ref{ap-fun} we present definitions and properties of the special functions relevant
in the uncoupled case, whereas in Appendix \ref{ap-finite} we provide details on the derivation
of the corresponding explicit expressions for $\cX_\gamma$ and the $\tau$ function.

\section{Complex HK metrics and Joyce structures from twistors}
\label{sec-HKmetric}

In this section, we review the construction of complex HK metrics from variations of BPS structures,
as well as the definition of the Joyce and Pleba\'nski potentials, and phrase them in the language of twistorial geometry.

\subsection{BPS structure}
Following \cite{Bridgeland:2016nqw} with minor deviations, we define
a BPS structure as a triplet $(\Gamma,Z,\Omega)$
consisting of a lattice $\Gamma\sim \IZ^{\oplus n}$
equipped with a skew-symmetric integer pairing $\langle \cdot,\cdot\rangle$, a linear map $Z:\Gamma\to \IC$
and a map $\Omega:\Gamma\backslash\{0\} \to \IZ[y,1/y]$ satisfying the following properties:
\begin{itemize}
\item[i)] Symmetry: $\Omega(-\gamma,y)=\Omega(\gamma,y)$ for all
$\gamma\in\Gamma$,

\item[ii)] Support property: $\Omega(\gamma,y)\neq 0 \Rightarrow |Z_\gamma| >  ||\gamma||$
where $||\cdot||$ is a fixed norm on $\Gamma\otimes\IR$,

\item[iii)] Convergence: for any $y\in\IC^\times$, there exists $R>0$ such that
$\sum_{\gamma\in\Gamma} |\Omega(\gamma,y)|\, e^{-R Z_{\gamma}}<\infty$,
\end{itemize}
where $Z_\gamma$ denotes the image of $\gamma$ under the map $Z$.
The BPS structure $(\Gamma,Z,\Omega)$ is uncoupled if it satisfies, in addition to i)--iii),
\begin{itemize}
\item[iv)]  $\langle \gamma,\gamma'\rangle\neq 0 \ \Rightarrow \ \Omega(\gamma,y)=0 \ \mbox{or} \
\Omega(\gamma',y)=0$
\end{itemize}

We consider families of  (convergent, refined) BPS structures, where $Z$ varies in an open set  $\cU$ inside
the space of stability conditions $\cS={\rm Stab}(\cD)$ for a CY3-category $\cD$ in the sense of
\cite{MR2373143}. To avoid cluttering, we omit the dependence of $\Omega(\gamma,y)$ on $Z$.
Across the locus of marginal stability, defined as the union of the
walls
\be
\cW_{\gamma_1,\gamma_2} = \{ Z\in \cS \ \mbox{such that}\  \Im(Z_{\gamma_1} \bZ_{\gamma_2})=0,
\  \Re(Z_{\gamma_1} \bZ_{\gamma_2})>0
\}
\ee
over all pairs of charges $\gamma_1,\gamma_2\in \Gamma$ such that $\langle \gamma_1,\gamma_2\rangle\neq 0$,
we require that the map $\Omega$ changes according to the Kontsevich-Soibelman
wall-crossing formula \cite{ks} (this condition is void in the uncoupled case).
We denote $\Gamma_{\!\star}=\Gamma\backslash\{0\} $ and
define $\Gamma_{\! +}$ as the set of vectors $\gamma\in\Gamma_{\!\star} $ such that
$\Im Z_\gamma>0$ or $Z_\gamma\in \IR^-$, and for simplicity we assume that $\Gamma_{\! +}$ stays constant throughout the open set $\cU$.
We define the rational DT
invariants $\bOm(\gamma,y)$ through the relation
\be
\bOm(\gamma,y) = \sum_{d|\gamma} \frac{y-1/y}{d(y^d-y^{-d})}\,\Omega(\gamma/d,y^d),
\label{bOm}
\ee
and denote by  $\Omega(\gamma)$ the value of
$\Omega(\gamma,y)$ at $y=1$.

In the physical context of four-dimensional field or string theories with $\cN=2$ supersymmetry,
$\Gamma$ is the lattice of electromagnetic charges, $\langle \cdot,\cdot\rangle$ is
the Dirac-Schwinger-Zwanziger pairing, $Z_\gamma$ is the central charge
and $\Omega(\gamma,y)$ are the refined BPS indices, or generalized Donaldson-Thomas invariants.
The condition i) is a consequence of CPT symmetry. The condition ii) is a technical condition
needed to ensure that the density of the set $\{ Z_\gamma\; :\; \gamma\in \Gamma,\; \Omega(\gamma)\neq 0\}$
inside $\IC$ does not grow exponentially at large $|Z|$ \cite{ks}.
The convergence property iii) is expected to hold in any non-gravitational supersymmetric field theory,
and in superstring theory or M-theory compactified on a CY threefold $\CY$ when it is non-compact,
but not in genuine string compactifications.
In the context of string vacua, the central charge $Z$ is restricted to a complex Lagrangian subspace of the
space ${\rm Stab}(\cD)$ of the relevant derived category of branes on $\CY$, determined locally by a prepotential.
The condition iv), also known as mutual locality, is satisfied
when there exists an electromagnetic
frame where all BPS dyons carry only electric charge. This is, in particular, the case
for type IIA strings compactified on a non-compact CY threefold without compact divisors.

\subsection{Twistor space}

Let us choose a basis of vectors $\gamma^a\in\Gamma_{\! +}$, $a=1,\dots, \dm$.
We denote by $q_a$ the coefficients of $\gamma$ in this basis,
i.e. $\gamma=\sum_{a=1}^{\dm}q_a\gamma^a$.
Then let $\cM$ be the holomorphic tangent bundle over $\cS$, parametrized by $(z^a,\theta^a)$ such that
the base is parametrized by $z^a=Z_{\gamma^a}$ and the fiber is parametrized by $\theta^a \p_{z^a}$.
In order to specify a complex HK metric on $\cM$, we need to specify a  holomorphic two-form
$\omega$ on the twistor space on $\cZ=\cM\times \IP^1_t$, viewed as a holomorphic fibration of $\cM$ over $\IP^1_t$,
and a family of sections (known as twistor lines) parametrized
by $\cM$ \cite{Hitchin:1986ea}.
More explicitly \cite{Alexandrov:2008ds},
we need to construct local  coordinates $\xi^a(z,\theta,t)$ on $\cZ$
such that  the  twistor lines are parametrized by $\xi^a(t)$ for $t\in \IC \cup \{\infty\}$ and fixed $(z,\theta)$,
and the holomorphic two-form takes the Darboux form
\be
\omega=\frac{t}{2}\sum_{a,b} \omega_{ab}\,  \de\xi^a \,\d\xi^b,
\ee
where $\omega_{ab}$ is the inverse of the constant (integer-valued)
skew-symmetric matrix $\omega^{ab}=\langle\gamma^a,\gamma^b\rangle$.
The HK two-forms $\omega_\pm,\omega_3$
on $\cM$ can then be read off
by pulling back $\omega$ from $\cZ$ to $\cM$ and identifying
\be
\omega= t^{-1} \omega_+ -\I  \omega_3 + t \omega_-\, .
\label{symform}
\ee

If one requires that the HK metric on $\cM$ descends to the torus bundle $\cM/\Gamma$, where
$\Gamma$ shifts the coordinates $\theta^a$, then $\Gamma$ should act
on $\cZ$ by integer shifts of the Darboux coordinates $\xi^a$, such that the holomorphic Fourier modes
$\cX_\gamma=\sigma_\gamma e^{2\pi\I \langle \gamma,\xi\rangle}$
descend to the quotient $\cZ/\Gamma$.
Here we included into $\cX_\gamma$ the factor known as quadratic refinement which is defined
as a map $\sigma:\Gamma\to\IC^\times$ such that
$\sigma_\gamma\, \sigma_{\gamma'} = (-1)^{\langle\gamma,\gamma'\rangle} \sigma_{\gamma+\gamma'}$.
Therefore, specifying $\xi^a$ is equivalent to specifying twisted
holomorphic Fourier modes $\cX_\gamma:\cZ\to \IC$ such that
\be
\cX_\gamma\cX_{\gamma'}=(-1)^{\langle\gamma,\gamma'\rangle}\cX_{\gamma+\gamma'}\ .
\ee

The holomorphic Fourier modes $\cX_\gamma$ are determined
as solutions of the following Riemann-Hilbert problem \cite{Bridgeland:2016nqw}.
In the limit $t\to 0$, we require that they reduce,
up to exponential corrections, to\footnote{The index `sf' stands for `semi-flat', since
the HK metric encoded by the twistor lines \eqref{cXsf} is flat along the fibers parametrized by $\theta^a$.
In fact, in contrast to the HK metrics appearing in the physical context {\it before} taking the conformal limit,
it is also flat along the base $\cS$ so that the full complex HK space encoded by \eqref{cXsf}
is simply $\IC^{2m}$ endowed with a flat metric of signature $(2m,2m)$.}
\be
\cXsf_\gamma=\sigma_\gamma\, e^{2\pi\I(\theta_{\gamma} - Z_\gamma/t)},
\qquad
\theta_\gamma=\langle \gamma,\theta\rangle.
\label{cXsf}
\ee
Away from $t=0$ and for any $\gamma'\in\Gamma_{\!\star}$ such that $\Omega(\gamma')\neq 0$, we require that
the Fourier modes $\cX_\gamma^\pm$ defined on the two sides of the BPS ray
(on clockwise and anticlockwise sides, respectively)\footnote{Our convention is opposite to the one in
\cite{Bridgeland:2016nqw,Bridgeland:2017vbr,Bridgeland:2019fbi},
which can be reached by flipping the sign of the skew-symmetric
pairing $\langle \gamma,\gamma'\rangle$. \label{foot-sign}}
\be
\ell_{\gamma'}= \{ t\in \IP^1:\quad Z_{\gamma'}/t\in \I\IR^-\} ,
\ee
are related by a KS symplectomorphism \cite{ks}
\be
\cX_\gamma^+ = \cX_\gamma^- (1-\cX_{\gamma'}^- )^{\Omega(\gamma')\,\langle\gamma,\gamma'\rangle}.
\ee
Finally, for large $|t|$ we require that $\cX_\gamma(t)$ behaves polynomially, i.e.
\be
|t|^{-k}<|\cX_\gamma(t)|<|t|^k
\ee
for some $k>0$.
This Riemann-Hilbert problem was  studied in detail in \cite{Bridgeland:2016nqw,Bridgeland:2017vbr}
for the case of uncoupled BPS structure, where explicit solutions can be written in terms of
some special functions which will be reviewed in section \ref{sec-local}.
When the BPS structure does not satisfy the above condition, i.e. it is not mutually local,
no explicit solutions are known at present.

As usual for such Riemann-Hilbert problems, it is convenient to reformulate
the previous boundary condition and discontinuity requirements as the
TBA-like system of integral equations \eqref{TBAeq}. These equations take a simpler form
when expressed in terms of the rational invariants \eqref{bOm} \cite{Filippini:2014sza}
\be
\cX_\gamma(t)=
\cXsf_\gamma(t)\,
\exp\[\frac{1}{2\pi\I}\sum_{\gamma'\in\Gamma_{\!\star}} \bOm(\gamma')
\langle\gamma,\gamma'\rangle\int_{\ell_{\gamma'}}\frac{\de t'}{t'}\,
\frac{t}{t'-t}\,\cX_{\gamma'}(t')\].
\label{TBAeq1}
\ee
This formulation allows in particular to generate a formal expansion
of $\cX_\gamma$ in powers of rational DT invariants, by replacing $\cX_{\gamma'}$ in the integral by
$\cXsf_{\gamma'}$ and then iterating this procedure as in \cite{Gaiotto:2008cd}. It also
makes it manifest that the solutions satisfy the homogeneity property (known as conformal invariance in this context)
\be
\label{confprop}
\left[  t\p_t  + \sum_a z^a \p_{z^a}   \right] \cX_\gamma=0 .
\ee

In fact, the above Riemann-Hilbert problem as well as the integral equation \eqref{TBAeq1} formally
arise as the conformal limit \cite{Gaiotto:2014bza,Filippini:2014sza} of the GMN construction \cite{Gaiotto:2008cd}.
In the latter,  the boundary condition \eqref{cXsf} has a more complicated form
$\cXsf_\gamma(t)= \sigma_{\gamma}e^{2\pi\I(\theta_{\gamma} - Z_\gamma/t+ \bZ_{\gamma}t)}$
while the kernel in the integral equation involves the factor $\hf\,\frac{t'+t}{t'-t}=\frac{t}{t'-t}+\frac12$
rather than $\frac{t}{t'-t}$ in  \eqref{TBAeq1}.
The conformal limit corresponds to $t\to 0$ and $z^a\to 0$ keeping the ratio $z^a/t$ fixed.
Whereas $\cXsf_\gamma$  trivially reduces to \eqref{cXsf} in this limit, in order
to arrive at  \eqref{TBAeq1} one has to absorb the contribution of the constant term $\frac12$
in the kernel into a redefinition of the coordinates  $\theta^a$
{\em before} taking the limit (a change of frame in the language of \cite{Filippini:2014sza}).
In particular, the new coordinates $\theta^a$ are no longer real.

The relation to the GMN construction just described allows us to translate various structures which have been developed
since its discovery to the current context. As we will see, some of them turn out to coincide with the principal objects
determining the geometry of the complex HK space $\cM$.

\subsection{Joyce connection and potential}
\label{sec_Joyce}

Given a solution to the system \eqref{TBAeq1}, we can construct a connection $\nabla$ on $\IP^1_t$
such that the holomorphic Fourier modes $\cX_\gamma$ are flat sections of this connection. It is shown in \cite{Filippini:2014sza}
that the resulting connection is analytic across the BPS rays, and has a double pole at $t=0$,
with coefficient given by the central charge function. Hence, it takes the form
\be
\nabla=\de-\(\frac{\hat Z}{t^2}+\frac{\hat f}{t}\)\de t,
\label{nabla}
\ee
where $\hat Z=\sum_a z^a \p_{\theta^a}$ and
$\hat f= \frac{1}{(2\pi)^2}\sum_{a,b} \omega^{ab} \p_{\theta^a} F \p_{\theta^b}$ is an Hamiltonian vector field
generated by a function $F$ on $\cM$, which we call the Joyce potential.
Put differently, the holomorphic Fourier modes satisfy
\be
\label{flatsec}
\left[ t \p_t - \frac{1}{t} \sum_a z^a \p_{\theta^a}
+\frac{1}{(2\pi)^2} \sum_{a,b} \omega^{ab} \p_{\theta^a} F \,\p_{\theta^b}
\right] \cX_\gamma =0,
\qquad
\forall\gamma\in\Gamma.
\ee
Note that the homogeneity property \eqref{confprop} requires $F$ to be homogeneous with respect to the coordinates $z^a$.

By construction, the Stokes data of this meromorphic connection is determined by the BPS indices,
which are locally independent of the coordinates $z^a$. It follows that the Fourier modes of the Joyce potential,
defined by\footnote{The zero mode of $F$ does not contribute to the vector field  $\hat f$ and may be set to 0.}
\be
F=\sum_{\gamma\in\Gamma_{\!\star} } F_\gamma\, e^{2\pi\I \theta_\gamma}
\label{defF}
\ee
satisfy the isomonodromy equation
\be
\de F_\gamma=\sum_{\gamma_1,\gamma_2\in \Gamma_{\!\star}
\atop \gamma_1+\gamma_2=\gamma}
\langle\gamma_1,\gamma_2\rangle\, F_{\gamma_1}F_{\gamma_2}\, \de\log Z_{\gamma_2}\, .
\label{isoeq}
\ee
In  \cite{joyce2007holomorphic}, Joyce has constructed a solution of this equation,
as  a formal series in powers of rational BPS indices,\footnote{The quadratic refinement
$\sigma_{\gamma_i}$ was not included in \cite{joyce2007holomorphic}, which relied on a slightly
different version of rational DT invariants, which lacks the property of being deformation invariant, see e.g. \cite{bridgeland2012introduction}  for details.}
\be
\label{FJoyce}
F_\gamma= \sum_{n=1}^\infty\sum_{\gamma=\sum_{i=1}^n \gamma_i} \, \frac{1}{2^{n-1}}
\left[  \sum_{\cT\in \mathbb{T}_n^\ell}\,
 \prod\limits_{e: i\to j}  \langle \gamma_i, \gamma_j \rangle \right]
\, J_n(Z_{\gamma_1},\dots, Z_{\gamma_n})\, \prod_{i=1}^n\, \sigma_{\gamma_i}\, \bOm(\gamma_i),
\ee
where the sum runs over all  decompositions of $\gamma$ into ordered sums of vectors $\gamma_i\in\Gamma_{\!\star}$,
 $J_n(z_1,\dots, z_n)$ are multi-valued, homogeneous
 functions on $(\IC^\times)^n$ and $\mathbb{T}_n^\ell$ is the
set of all connected trees with $n$ vertices labeled from 1 to $n$, with edges oriented as $e:i\to j$ whenever $i<j$
(The number of such trees, which we call unrooted labeled trees for brevity, is $n^{n-2}$).
The formal series \eqref{FJoyce} is smooth across walls of marginal stability and satisfies \eqref{isoeq}
provided the BPS indices jump according to the KS formula
and the functions $J_n(z_1,\dots, z_n)$ satisfy the following axioms:
\begin{subequations}
\begin{itemize}
\item[a)] $J_n(z_1,\dots, z_n)$ is a homogeneous function of its arguments,
holomorphic (and continuous) on the domain
\be
\label{defDn}
\cD_n = \{ (z_1,\dots, z_n) \in (\IC^{\times})^n\ :\quad z_{k+1}/z_k\notin [0,\infty] \ \mbox{for all}\ 1\leq k<n\}.
\ee
\item[b)] It satisfies  the following recursion relations, with $J_1(z)=1/(2\pi\I)$:
\be
\label{dFJoyce}
\de J_n = \sum_{k=1}^{n-1}\, J_k(z_1,\dots z_k)\, J_{n-k}(z_{k+1},\dots, z_n)\,
\left[ \frac{\de z_{k+1}+\dots+ \de z_n}{z_{k+1}+\dots+z_n} -
\frac{\de z_{1}+\dots+ \de z_k}{z_{1}+\dots+z_k}\right].
\ee
\item[c)] For all $(z_1,\dots, z_n)\in (\IC^\times)^n$,
\be
\Sym J_n(z_1,\dots z_n) = 0,
\label{vanishsym}
\ee
where $\Sym$ denotes symmetrization over all permutations of indices of arguments.
\item[d)]
It satisfies the growth condition as $z_k\to 0$ keeping $z_{\ell\neq k}$ fixed,
\be
\label{growth}
J_n(z_1,\dots z_n) = o(1/|z_k|).
\ee
\label{axiomsJ}
\end{itemize}
\end{subequations}
\vspace{-0.8cm}
As shown in \cite{joyce2007holomorphic}, these axioms specify the functions $J_n$ uniquely. In particular,
\be
(2\pi\I)^2 \, J_2(z_1,z_2)=
 \begin{cases}
\log\left({z_2}/{z_1} \right)-\I\pi,\quad
& z_2/z_1 \notin [0,\infty],
\\
\log\left({z_2}/{z_1} \right),
& z_2/z_1 \in [0,\infty],
\end{cases}
\label{J2}
\ee
where $\log z$ is defined such that $\Im\log z\in [0,2\pi)$.

\subsection{Pleba\'nski potential and heavenly metrics}

The Pleba\'nski potential is obtained from the Joyce potential by dividing
each Fourier mode by the corresponding central charge,
\be
W=\frac{1}{2\pi\I}\sum_{\gamma\in \Gamma_{\!\star} }
\frac{F_\gamma}{Z_\gamma} \, e^{2\pi\I \theta_\gamma} \,.
\label{defW}
\ee
Conversely, $F=\sum_{a=1}^{\dm} z^a\p_{\theta^a} W$.
Since $F_\gamma$ are homogenous functions of the coordinates $z^a$, $W$ is a homogeneous function of degree $-1$.
It follows from \eqref{isoeq} that
$W$ satisfies the heavenly equation
\be
\frac{\p^2 W}{\p z^a\p \theta^b}-\frac{\p^2 W}{\p z^b\p \theta^a}
=\frac{1}{(2\pi)^2}\sum_{c,d=1}^{\dm}
\omega^{cd} \, \frac{\p^2 W}{\p \theta^a\p\theta^c}\, \frac{\p^2 W}{\p \theta^b\p\theta^d}\, .
\label{heaveq1}
\ee
As explained in the introduction, this is also the integrability condition for the vector fields
$u_a$ in \eqref{actonX}. In fact, the horizontal sections for the vector fields are the same
ones as for the connection \eqref{flatsec}, namely
\be
\left[ \frac{\p}{\p z^a}+\frac{1}{t}\, \frac{\p}{\p\theta^a}
-\frac{1}{(2\pi)^2}\sum_{b,c=1}^{\dm}\omega^{bc}\frac{\p^2 W}{\p \theta^a\p\theta^b}\,\frac{\p}{\p\theta^c}
\right] \cX_\gamma =0.
\label{actonX1}
\ee
Indeed, using \eqref{confprop} it is easy to check that \eqref{actonX1} implies \eqref{flatsec}.

As explained in \cite{Bridgeland:2020zjh,Dunajski:2020qhh}, a solution of \eqref{heaveq1}
defines a complex \hk metric on $\cM$, the total space of the holomorphic tangent bundle
over $\cS$,
\bea
\label{defg}
g
&=& \sum_{a,b=1}^{\dm} \omega_{ab}\(v^a\otimes h^b+h^b\otimes v^a\),
\eea
where $h^a$, $v^a$ are the horizontal and vertical one-forms,
\be
h^a=\de z^a,
\qquad
v^a=\de\theta^a-\frac{1}{(2\pi)^2}\sum_{c,d=1}^{\dm}\omega^{cd}\frac{\p^2 W}{\p \theta^a\p\theta^c}\,\de z^d.
\ee
The dual vector fields
\be
h_a=\frac{\p}{\p z^a}-\frac{1}{(2\pi)^2}\sum_{c,d=1}^{\dm}\omega^{cd}\frac{\p^2 W}{\p \theta^a\p\theta^c}\,\frac{\p}{\p\theta^d},
\qquad
v_a=\frac{\p}{\p\theta^a}
\ee
satisfy $(h_a,h^b)=(v_a,v^b)=\delta_a^b$, $(v_a,h^b)=(h_a,v^b)=0$
and correspond to the horizontal and vertical lift of the vector field $\p_{z^a}$ on $S$.
Importantly the vector fields $h_a$ commute among each other by virtue of \eqref{heaveq1}
(or rather, the $\theta^a$ derivative thereof), just like the $v_a$'s.
In contrast, the commutator $[v_a,h_b]$ is non zero, but the vector fields $u_a=h_a+t^{-1} v_a$
appearing in \eqref{actonX1} do commute for any $t$.

In the basis $(v_a,h_a)$,
the three complex structures on $\cM$ are defined by the block matrices
(adapting the result of \cite{Bridgeland:2020zjh,Dunajski:2020qhh} to our conventions)
\be
J_1=\(\begin{array}{cc}
0 & -\unit \\ \unit & 0
\end{array}\)\,
\qquad
J_2=\(\begin{array}{cc}
0 & -\I\unit \\ -\I\unit & 0
\end{array}\),
\qquad
J_3=\(\begin{array}{cc}
\I\unit & 0 \\ 0 & -\I\unit
\end{array}\).
\ee
The associated symplectic forms $\omega_i(X,Y)=g(X,J_i Y)$, or equivalently their complex combinations
$\omega_\pm=\hf\(\omega_1\mp\I\omega_2\)$, are given by
\be
\omega_3=-\I\sum_{a,b}\omega_{ab} v^a\wedge h^b,
\qquad
\omega_+=\hf\sum_{a,b}\omega_{ab} h^a\wedge h^b,
\qquad
\omega_-=\hf\sum_{a,b}\omega_{ab} v^a\wedge v^b ,
\ee
so that the holomorphic two-form \eqref{symform} can be written as
\be
\omega=\frac{1}{2t}\sum_{a,b}\omega_{ab} (h^a-t v^a)\wedge (h^b-tv^b).
\ee
This is consistent with the identification $\de\xi^a=v^a-h^a/t$.
Note that $\omega_+$ is the pull-back of the symplectic form $\hf\,\omega_{ab} \de z^a \de z^b$ on $S$,
hence vanishes on the vertical space of the fibration $\cM\to S$, while  $\omega_-$ is non-degenerate along the fiber.
Since $W$ is homogeneous of degree $-1$, the metric \eqref{defg} has a homothetic Killing vector $z^a \p_{z^a}$,
hence is a complex HK cone.

\section{Integral representations and formal expansions}
\label{sec_joyce}

In this section we establish the main results of this work, namely,
the integral expressions of the Pleba\'nski potential $W$ \eqref{Wfull}
and the Joyce potential \eqref{Ffull},
and compare the formal expansion of these functions in powers of rational DT invariants with the original
construction of Joyce \cite{joyce2007holomorphic}.
Finally, we discuss a $\tau$ function that can be defined in the uncoupled case.

\subsection{Integral representations\label{sec_int}}

Here, our goal is to show that the integral expressions  \eqref{Wfull}, \eqref{Ffull} for $W$ and $F$
satisfy the section condition \eqref{actonX1}, the heavenly equation \eqref{heaveq1}, and the isomonodromy equation \eqref{isoeq}.
In fact, it is sufficient to prove only the section condition \eqref{actonX1}, since the two other equations follow by construction.
Nonetheless, for completeness we also provide their proofs in Appendices \ref{ap-proofHE} and \ref{ap-proofISO}, respectively.

To avoid cluttering, it will be useful to introduce some shorthand notations,
\be
\label{short}
\cX_i=\cX_{\gamma_i}(t_i),
\qquad
K_{ij}=\frac{\langle\gamma_i,\gamma_j\rangle}{2\pi\I}\,\frac{t_i t_j}{t_j-t_i}\, ,
\qquad
\IS_i=\sum_{\gamma_i\in\Gamma_{\!\star} }
\bOm(\gamma_i)\int_{\ell_i}\frac{\de t_i}{t_i^2}\, .
\ee
Note that  the kernel $K_{ij}$
is symmetric under $i\leftrightarrow j$.
With these notations, the TBA equation \eqref{TBAeq1} becomes
\be
\cX_1=\cX_1^{\rm sf}
\exp\[\IS_2 K_{12} \cX_2 \],
\quad
\label{TBAeq-sh}
\ee
where we recall that $\cX_\gamma^{\rm sf} = \sigma_\gamma e^{2\pi\I(\theta_{\gamma}-Z_\gamma/t)}$,
whereas the integral representations \eqref{Wfull} and \eqref{Ffull}, after rewriting them in terms of rational invariants,
take the form
\bea
\label{Wfull-short}
W&=&\frac{1}{2\pi\I}\, \IS_1 \cX_1\( 1-\hf\, \IS_2 K_{12} \cX_2\),
\\
F&=&\IS_1 Z_{\gamma_1}\cX_1.
\label{Ffull-short}
\eea

As a first step, let us evaluate derivatives of the Pleba\'nski potential.
This can be done using the relation
\be
\p_{\theta^a} \cX_1=\cX_1
\(2\pi\I q_{1,a}  +\IS_2 K_{12} \p_{\theta^a} \cX_2 \),
\label{derX-sh}
\ee
which follows by differentiating the equation \eqref{TBAeq-sh}.
It allows to find
\bea
\p_{\theta^a} W&=&\frac{1}{2\pi\I}\, \IS_1 \p_{\theta^a} \cX_1 \( 1-\IS_2 K_{12} \cX_2\)
 \nn\\
&=&\frac{1}{2\pi\I}\[ \IS_1\cX_1\(2\pi\I q_{1,a} + \IS_2 K_{12} \p_{\theta^a} \cX_2\)
- \IS_1 \p_{\theta^a} \cX_1 \IS_2 K_{12} \cX_2\]
\nn\\
&=&\IS_1\cX_1 q_{1,a},
\label{pW}
\eea
In the same way, one obtains
\be
\p_{z^a} W=-\IS_1\cX_1\, \frac{q_{1,a}}{t_1}\, ,
\ee
and therefore
\be
\p_{\theta^a} \p_{\theta^b} W= \IS_1 \p_{\theta^a} \cX_1\, q_{1,b},
\qquad
\p_{z^a}\p_{\theta^b} W =-\IS_1 \p_{\theta^b} \cX_1 \, \frac{q_{1,a}}{t_1}\, .
\label{p2W}
\ee
In particular, combined with the relation \eqref{FW},
the result \eqref{pW} immediately leads to the representation \eqref{Ffull-short}
of the Joyce potential $F$.

Next, we substitute these results into \eqref{actonX1} and evaluate the derivatives of the Fourier mode appearing there.
Unfortunately, the only way we have found to recombine these derivatives is to represent them as
infinite series obtained by iterating the relation \eqref{derX-sh}.\footnote{Note that this iteration procedure
is distinct from the usual iteration of the TBA equations \eqref{TBAeq1} which leads to a formal
series in powers of $\cXsf_{\gamma}$. This second expansion is the subject of the next subsection.}
Assuming convergence of such a series, the first two terms in \eqref{actonX1} are then found to be
\bea
\(\p_{z^a}+t_0^{-1}\p_{\theta^a}\)\log\cX_0&=&2\pi\I\sum_{n=1}^\infty
\prod_{i=1}^n \( \IS_i K_{i-1,i} \cX_i \)q_{n,a} \, \frac{t_n-t_0}{t_0  t_n}\, .
\label{1term-all}
\eea
On the other hand, the last term gives
\bea
&&
-\frac{1}{(2\pi)^2}\sum_{b,c}\omega^{bc}\frac{\p^2 W}{\p \theta^a\p\theta^b}\,\frac{\p}{\p\theta^c}\log\cX_0
= \frac{1}{(2\pi\I)^2}\,\IS_1\p_{\theta^a} \cX_1 \sum_{b,c}\omega^{bc}q_{1,b}\, \p_{\theta^c} \log\cX_0
\nn\\
&=& -\sum_{n'=0}^\infty \prod_{i=1}^{n'} \( \IS_{i} K_{i-1,i} \cX_{i} \)
\IS_{n'+1}\cX_{n'+1} \langle \gamma_{n'}, \gamma_{n'+1} \rangle
\sum_{n''=1}^\infty \prod_{i=n'+2}^{n'+n''} \( \IS_{i} K_{i-1,i} \cX_{i} \)q_{n'+n'',a}
\nn\\
&=& -2\pi\I \sum_{n=1}^\infty \prod_{i=1}^{n} \( \IS_i K_{i-1,i} \cX_i \)q_{n,a}
\sum_{i=0}^{n-1}\frac{t_{i+1}-t_i}{t_it_{i+1}}\, .
\eea
where in the second line the first sum comes from the expansion of $ \p_{\theta^c} \log\cX_0$, whereas
the second one originates from $\p_{\theta^a} \cX_1$, and in the last line $n=n'+n''$.
Since $\sum_{i=0}^{n-1}\frac{t_{i+1}-t_i}{t_it_{i+1}}=\frac{t_n-t_0}{t_0 t_n}$,
this contribution reproduces \eqref{1term-all} with the opposite sign so that \eqref{actonX1}
is indeed satisfied. $\Box$

The heavenly and isomonodromy equations, \eqref{heaveq1} and \eqref{isoeq}, are proven using similar manipulations.
We relegate these proofs to Appendices \ref{ap-proofHE} and \ref{ap-proofISO}.

Importantly, as explained in the Introduction,
the integral representations \eqref{Wfull-short} and \eqref{Ffull-short} allow to identify
the Pleba\'nski potential $W$ with the conformal limit
of the instanton generating potential of \cite{Alexandrov:2018lgp}, and the Joyce potential $F$ with the conformal limit
of the contact potential \cite{Alexandrov:2008nk,Alexandrov:2009zh}, or equivalently the
Witten index on $\IR^3$ \cite{Alexandrov:2014wca},
thereby providing a possible physical interpretation.

\subsection{Formal expansions}
\label{sec_asym}

As explained in \cite[\S C]{Gaiotto:2008cd}, the integral equations of the type \eqref{TBAeq1}
can be viewed as a fixed point problem, and solved by starting with the `approximate'
solution $\cX_\gamma(t)\sim
\cXcl_\gamma(t)$ and then iterating. This results in a formal expansion in powers of
BPS indices given by
\cite{Filippini:2014sza}
\be
\label{Hexpand}
\cX_{\gamma_1}(t_1) = \Hcl_{\gamma_1}(t_1) \, \sum_{n=1}^{\infty} \left(
\prod_{i=2}^n  \sum_{\gamma_i\in\Gamma_{\!\star}} \bOm(\gamma_i)
\int_{\ell_{\gamma_i}} \frac{\de t_i}{t_i^2}\,  \Hcl_{\gamma_i}(t_i) \right) \sum_{\cT\in \IT_n^{\rm r}}\,
\frac{\prod_{e:i\to j}K_{ij}
}{|{\rm Aut}(\cT)|}\, ,
\ee
where $\IT_n^{\rm r}$ is the set of rooted trees with $n$ vertices, and each vertex
of the tree is decorated with a charge $\gamma_i\in\Gamma_{\!\star} $, with  $\gamma_1$ associated
to the root vertex. More explicitly, using the same short hand notation as in the previous subsection
we get at the first few orders
\bea
\label{Xexpandshort}
\cX_{1} &= &  \cXcl_{1} + \IS_2 K_{12} \cXcl_{1+2}
+  \IS_2 \IS_3 \left( \tfrac12\, K_{12} K_{13} + K_{12} K_{23} \right)
 \cXcl_{1+2+3}
 \\
& + &  \IS_2\IS_3\IS_4\left( \tfrac16\, K_{12} K_{13} K_{14}
+ \tfrac12\, K_{12} K_{23} K_{24} + K_{12} K_{13} K_{24} + K_{12} K_{23} K_{34} \right)
 \cXcl_{1+2+3+4}+ \dots
\nn
\eea

The asymptotic expansion of $F$ follows from \eqref{Ffull-short} by integrating over $t_1$,
and symmetrizing over the $\gamma_i$'s,
\be
F = \sum_{n=1}^{\infty} \prod_{i=1}^n \left(
\sum_{\gamma_i\in \Gamma_{\!\star}} \bOm(\gamma_i)
\int_{\ell_{\gamma_i}}\frac{\de t_i}{t_i^2}\, \cX^{\rm sf}_{\gamma_i}(t_i) \right)\ Z_{\sum_i \gamma_i}\,
\sum_{\cT\in\mathbb{T}_n} K_\cT(\{\gamma_i,t_i\}),
\label{funF}
\ee
where $\cT$ runs over unrooted trees with arbitrary (but fixed) labeling of vertices and\footnote{Alternatively,
one can sum over unrooted labeled tress,
in which case the coefficient $1/|{\rm Aut}(\cT)|$ in $K_\cT$ is replaced by $1/n!$.}
\be
\label{defKT}
K_\cT(\{\gamma_i,t_i\}) = \frac{1}{(2\pi\I)^{n-1}|{\rm Aut}(\cT)|}
\prod_{e:i\to j} \frac{\langle\gamma_{i},\gamma_{j}\rangle \,t_i t_j}{t_j-t_i}\, .
\ee
More explicitly, at the first few orders we get
\bea
\label{Fexpandshort}
F  &= & \IS_1 Z_1  \cXcl_{1} + \frac12 \IS_1  \IS_2 Z_{1+2}  K_{12} \cXcl_{1+2}
+ \frac12 \IS_1  \IS_2 \IS_3  Z_{1+2+3} K_{12} K_{13}
 \cXcl_{1+2+3}
\nn \\
& + &  \IS_1   \IS_2\IS_3\IS_4 Z_{1+2+3+4} \left( \tfrac16\, K_{12} K_{13} K_{14}
+ \tfrac12\, K_{12} K_{23} K_{34} \right)
 \cXcl_{1+2+3+4}+ \dots
\eea
Further using \eqref{defW}, we find that  $W$ is given by a similar sum over unrooted, unlabeled trees,
without the factor of $Z_{\gamma}$ in each term,
\be
W =
\frac{1}{2\pi\I} \sum_{n=1}^{\infty} \prod_{i=1}^n \left(
\sum_{\gamma_i\in \Gamma_{\!\star}} \bOm(\gamma_i)
\int_{\ell_{\gamma_i}}\frac{\de t_i}{t_i^2}\, \cX^{\rm sf}_{\gamma_i}(t_i) \right)
\sum_{\cT\in\mathbb{T}_n} K_\cT(\{\gamma_i,t_i\}).
\label{iterW}
\ee

Let us consider now the Fourier modes $F_\gamma$ of the Joyce potential defined as in \eqref{defF}.
Collecting the coefficient of $e^{2\pi\I\theta_\gamma}$ in \eqref{funF}, we find
\be
\label{Fours}
F_\gamma = Z_{\gamma}\sum_{n=1}^{\infty}
\sum_{\gamma=\sum_{i=1}^n \gamma_i} \sum_{\cT\in\mathbb{T}_n}
\prod_{i=1}^n  \(\sigma_{\gamma_i}\bOm(\gamma_i)\int_{\ell_{\gamma_i}}\frac{\de t_i}{t_i^2}\,e^{-2\pi\I z_i/t_i}\)
K_\cT(\{\gamma_i,t_i\}),
\ee
where the sum runs over all ordered decompositions of $\gamma$.
While this expansion looks broadly similar to \eqref{FJoyce},
it is also markedly different: in contrast to \eqref{Fours},
in \eqref{FJoyce} the dependence on all pairings $\langle \gamma_i, \gamma_j \rangle$ is factorized.
Nonetheless, we claim that the two expressions are in fact identical. The agreement follows from the following  conjecture:
\begin{conj}
\label{conj-KT}
For each unrooted, unlabeled tree $\cT$ one has
\be
\Sym K_\cT(\{\gamma_i,t_i\})=
\frac{1}{(4\pi\I)^{n-1}}\,\Sym \Biggl[\(\sum_{\rm labels }\prod_{e:i\to j} \langle \gamma_i, \gamma_j \rangle\) \prod_{i=1}^{n-1}
\frac{t_i t_{i+1}}{t_{i+1}-t_i}\Biggr],
\label{propint}
\ee
where on the r.h.s. the sum goes over all labelings $\ell:v\mapsto i\in[1,n]$
of the vertices of $\cT$, and the edges are oriented from $v$ to $v'$ whenever $\ell(v)<\ell(v')$.
\end{conj}
This claim is established for $n\leq 4$ in Appendix \ref{ap-KT}, and has been checked for random values of
the $t_i$'s and $\langle \gamma_i, \gamma_j \rangle$'s up to $n=7$ using a computer.

Granting this fact, \eqref{propint}
allows to rewrite \eqref{Fours} into the same way as \eqref{FJoyce}, with
$J(z_1,\dots, z_n)$ given by
\be
\label{Jours}
J_n(z_1,\dots, z_n) =
\frac{z_1+\dots+z_n}{(4\pi\I)^{n-1}} \( \prod_{i=1}^n  \int_{\ell_{\gamma_i}}\frac{\de t_i}{t_i^2} \,e^{-2\pi\I z_i/t_i}\)
\prod_{i=1}^{n-1} \frac{t_i t_{i+1}}{t_{i+1}-t_i}\, .
\ee
In Appendix \ref{sec_JJt}, we prove that \eqref{Jours} satisfies the axioms \eqref{axiomsJ}.
Hence, our function $F$ coincides with the solution of the isomonodromy equation given by Joyce.

We note also that a nice representation of the multiple integral in \eqref{Jours}
is obtained by the change of variables $t_i=\I z_i/(s u_i)$ with $s, u_i\in \IR^+$ and $\sum_{i=1}^n u_i=1$.
It allows to rewrite the function $J_n$ as
\bea
J_n(z_1,\dots, z_n) &=&- \frac{\I \sum_i z_i}{(2\pi\I)^{n-1} \prod_{i=1}^n z_i}
\int_0^\infty \de s\, e^{-2\pi s}\,
\int_{0\leq w_1\leq 1\atop \sum_{i=1}^n u_1=1} \de u_1\cdots \de u_{n-1}\,
\prod_{i=1}^{n-1} \frac{z_i z_{i+1}}{z_{i+1} u_i-z_i u_{i+1}}
\nn\\
&=&\frac{\sum_i z_i}{(2\pi\I)^n\prod_{i=1}^n z_i}
\int_{0\leq u_1\leq 1\atop \sum_{i=1}^n u_1=1}
\frac{ \de u_1\cdots \de u_{n-1}}
{ \prod_{i=1}^{n-1} \left( \frac{u_i}{z_i} - \frac{u_{i+1}}{z_{i+1}}\right)}\, .
\label{defJtn}
\eea
Upon setting $v_k=\sum_{i=1}^k u_i$, this may be rewritten as an integral over the simplex
$0<v_1<v_2<\dots<v_{n-1}<1$, hence as an iterated integral in the language of \cite{chen1977iterated}.

Yet another interesting representation is obtained by a further change of variables $w_{i}=z_{i+1} u_i-z_i u_{i+1}$.
One can show that it produces the Jacobian\footnote{We remark that this change of variables can be done already in \eqref{Fours}, i.e.
it has a generalization to any tree $\cT$. In that case one defines $w_{ij}=z_{j} u_i-z_i u_{j}$
for any edge $e:i\to j$. Then the Jacobian is given by $\(\sum_i z_i\)  \frac{\prod_{e:i\to j}z_i z_j}{\prod_i z_i}$
and leads to the integral $\int_{R_\cT} \prod_{e:i\to j}  \frac{\de w_{ij}}{w_{ij}}$.
The results presented in \eqref{Jacob} and \eqref{integral-w} correspond to the particular case of the tree of
the simplest topology $\bullet\!\mbox{---}\!\bullet\!\mbox{--}\cdots \mbox{--}\!\bullet\!\mbox{---}\!\bullet\,$.
}
\be
\frac{\p\{w_{i}\}}{\p \{u_i\}}=\(\sum_{i=1}^n z_i\) \prod_{i=2}^{n-1} z_i  ,
\label{Jacob}
\ee
which cancels all $z_i$-dependent factors in \eqref{defJtn}.
Therefore, the integral can be recast as
\be
J_n(z_1,\dots, z_n)
= \frac{1}{(2\pi\I)^n} \int_{R(z_1,\dots, z_n)} \prod_{i=1}^{n-1}  \frac{\de w_{i}}{w_{i}}
\label{integral-w}
\ee
where the domain $R(z_1,\dots, z_n)\subset (\IC^\times)^n$ is the image of the
simplex $\{0<u_i<1, \sum_{i=1}^n u_i=1$\} under the map $\{u_i\}\mapsto\{w_i\}$ (notice that
$w_i$ does not vanish as long as the $z_i$'s belong
to the domain \eqref{defDn}). Thus $J_n(z_1,\dots, z_n)$ is simply the volume of the image of
$R$ under $w_i\mapsto\log w_i$. Unfortunately, we have not found a more direct way
to describe the domain $R$.

Finally, we compare \eqref{Jours} with the solution of the same set   of axioms found in  \cite{bridgeland2013stokes}
by expressing the connection $\nabla$ in terms of its Stokes data. The latter is given by
\be
\label{defJn}
\tilde J_n(z_1,\dots, z_n) = \frac{1}{(2\pi\I)^n} \sum_{\cT\in \mathbb{T}_n^{\rm pl}}  (-1)^{|V_\cT|} \prod_{v\in V_{\cT}}
M(z \vert_ {\rm ch(v)}),
\ee
where $\cT$ runs over rooted planar trees (the number of such trees is the Catalan number $C_{n-1}$),
with vertices $v$ decorated by linear combinations $z_v$ of $z_i$ such that
the root has $z_v=z_1+\dots +z_n$, the leaves of the tree have $z_v=z_i$ and $z_v=\sum_{v'\in \ch(v)} z_{v'}$
at each vertex (where $v'$ runs over the children of $v$). The factor $M(z \vert_ {\rm ch(v)})$ is
defined by evaluating the Goncharov polylogarithm \cite{Goncharov:1998kja}
\bea
\label{defMn}
M_n(z_1,\dots, z_n)
&=&  \int_{0\leq t_1\leq t_2 \leq \dots \leq t_{n-1}\leq s_n}
 \frac{\de t_1}{t_1-s_1} \wedge  \dots \wedge \frac{\de t_{n-1}}{t_{n-1}-s_{n-1}},
\eea
where $s_i=z_1+\dots+z_i$,
on the ordered set $\{z_{v'} \in \ch(v) \}$.
Since the solution to the conditions a)--d) is known to be unique \cite{joyce2007holomorphic},
the functions $\tilde J_n$ should coincide with the functions $J_n$ found above.
However, while the iterated integrals \eqref{defMn} are very similar to \eqref{defJtn}, it is unclear to us
how the various trees in \eqref{defJn} conspire to produce the much simpler result  \eqref{defJtn}.

\subsection{$\tau$ function}

The $\tau$ function was introduced in \cite{Bridgeland:2016nqw} in the case of uncoupled BPS structures,
as a solution of the following set of equations
\be
\sum_{b=1}^{\dm} \omega^{ab} \p_{z^b} \log\tau =\p_t \log\(\cX_{\gamma^a}/\cXsf_{\gamma^a}\),
\qquad
\forall a=1,\dots, \dm.
\label{eqtau}
\ee
It was explicitly found in a few cases where the solution to the Riemann-Hilbert problem was known,
but only for $\theta^a=0$.
Moreover, it was unclear why this function should exist and whether this definition makes sense more generally.
Our approach based on the TBA-like equations allows to address all these questions.

Evaluating the r.h.s. of \eqref{eqtau} as in \eqref{1term-all}, one finds that
the definition of the $\tau$ function requires the following relation
\be
\p_{z^a} \log\tau= \IS_1 \frac{q_{1,a} \, t_1}{t_1-t}\, \cX_1
\(\frac{Z_1}{t_1}+ \sum_{n=2}^\infty\prod_{i=2}^n\IS_i K_{i-1,i}\cX_i\, \frac{Z_i}{t_i}\).
\label{dertau}
\ee
It implies the integrability condition
\be
\label{taucond}
\IS_1  \frac{t_1 }{t_1-t} (q_{1,a}\, \p_{z^{b}}-q_{1,b}\, \p_{z^{a} })
\[\cX_1
\(\frac{Z_1}{t_1}+ \sum_{n=2}^\infty\prod_{i=2}^n\IS_i K_{i-1,i}\cX_i\, \frac{Z_i}{t_i}\)\]=0.
\ee
This condition is trivially satisfied in the mutually local case since the derivative $\p_{z^{b}}$ produces the factor $q_{1,b}$.
In that case, \eqref{dertau} reduces to
\be
\label{taumut}
\frac{\p\log\tau}{\p z^a}= \sum_{\gamma\in\Gamma_{\!\star}}\sigma_\gamma\bOm(\gamma)\int_{\ell_\gamma}
\frac{\de t'}{t'^2}\,\frac{q_{1,a} \,Z_\gamma}{t'-t}\,
e^{2\pi\I(\theta_\gamma-Z_\gamma/t')}.
\ee
which can be easily integrated. The result is given by
\be
\log\tau =
\frac{1}{4\pi^2}\sum_{\gamma\in\Gamma_{\!\star}}\sigma_\gamma\bOm(\gamma)\, e^{2\pi\I\theta_\gamma}
\[\int_{\ell_\gamma} \frac{\de t'}{t'} \,\frac{t}{t'-t}
\(1+2\pi\I\, \frac{Z_\gamma}{t'}\)e^{-2\pi\I Z_\gamma/t'}-\log (Z_\gamma/t)\].
\label{restau}
\ee
It provides a generalization
of the previous results \cite{Bridgeland:2016nqw,Bridgeland:2017vbr} to the case of non-vanishing $\theta^a$.
Note that $\log\tau$ is defined only up to the addition of an arbitrary function of $\theta^a$.
This fact will be important in the discussion of concrete examples in section \ref{sec-local}.

In contrast, in the mutually non-local case, the derivative $\partial_{z^b}$
acts on all $\cX_i$'s in the last term, producing the charges $q_{i,b}$.
Because the unintegrated variable $t$ is present only in the denominator involving $t_1$,
there is no symmetry between the label 1 and the labels $i$ which can ensure the integrability condition.
Moreover, one can expand the l.h.s. of \eqref{taucond} in powers of BPS indices as in \S\ref{sec_asym} and
check explicitly that the quadratic term is non-vanishing.
Thus, the integrability condition is no longer satisfied, which shows that the definition \eqref{eqtau}
of $\tau$ function does not extend beyond the case of uncoupled BPS structures.

Nevertheless, the results of \cite{Alexandrov:2017qhn} indicate
that such an extension might exist provided one suitably modifies
the defining differential equation. Indeed, in \cite[Eq.(3.5)]{Alexandrov:2017qhn}
a function $\tcJ$ was defined which by applying appropriate (non-linear) operators,
allowed to obtain all Darboux coordinates on the twistor space describing a closely related moduli space,
in the two-instanton approximation (i.e. at second order in powers of rational BPS indices).
This function can be seen as a $t$-dependent version of the instanton generating potential which, as we have explained,
coincides with the Pleba\'nski potential $W$.
However, in contrast to that potential, its generalization at higher orders is not known yet.

\section{The case of uncoupled finite BPS structures}
\label{sec-local}

In this section, we consider the case of
uncoupled BPS structures, such that $\langle \gamma,\gamma'\rangle=0$
whenever $\Omega(\gamma)$ and $\Omega(\gamma')$ are both non-zero.
Furthermore, we assume that only a finite number of BPS indices $\Omega(\gamma)$
are non-vanishing (leading to an infinite number of non-zero rational indices $\bOm(\gamma)$).

Let us start with the Pleba\'nski potential.
For a finite uncoupled BPS structure, its integral representation \eqref{Wfull-short} reduces to
\be
W=\frac{1}{2\pi\I} \sum_{\gamma\in\Gamma_{\!\star} }\sigma_\gamma \bOm(\gamma)
\int_{\ell_{\gamma}}\frac{\de t}{t^2}\,e^{2\pi\I(\theta_{\gamma} - Z_{\gamma}/t)} .
\label{Wfinite-int}
\ee
Changing the integration variable to $s=\I Z_\gamma/t$ so that the integration contour coincides with $\IR^+$,
one obtains a very simple result
\be
W=\frac{1}{(2\pi\I)^2} \sum_{\gamma\in\Gamma_{\!\star} } \frac{\sigma_\gamma\bOm(\gamma)}{Z_\gamma}\, e^{2\pi\I \theta_{\gamma}}
=\frac{1}{(2\pi\I)^2} \sum_{\gamma\in\Gamma_{\!\star} } \frac{\Omega(\gamma)}{Z_\gamma}\, \Li_3\(e^{2\pi\I \vth_{\gamma}}\),
\ee
where in the second equality we expressed the result in terms of the variables $\vth_\gamma$ defined by
\be
e^{2\pi\I \vth_\gamma}=\sigma_{\gamma}\, e^{2\pi\I \theta_\gamma},
\label{defvth}
\ee
and integer valued BPS indices so that the remaining sum is finite.
Due to the symmetry property of the BPS structure, one may
combine the contributions of the opposite vectors $\pm\gamma$ using
the relation \eqref{Li-ident},
arriving at
\be
W=-\frac{\pi\I}{3}\sum_{\gamma\in\Gamma_{\!+} } \frac{\Omega(\gamma)}{Z_\gamma}\, B_3([\vth_\gamma]),
\label{Wfinite}
\ee
where $B_3(x)$ is the Bernoulli polynomial \eqref{Berpol}. The bracket notation $[\theta]$ is defined in \eqref{defbr}
and its presence reflects the periodicity in the real part of $\theta^a$ inherent to the original representation \eqref{Wfinite-int}.
Note that due to $[-\theta]=1-[\theta]$ and $[\hf-\theta]=1-[\hf+\theta]$,
as well as the symmetry property of the Bernoulli polynomial, the Pleba\'nski potential is automatically odd
in $\theta^a$ for arbitrary quadratic refinement, as required in \cite{Bridgeland:2019fbi}.

Starting with \eqref{Ffull-short} and performing exactly the same manipulations, one can also derive
similar results for the Joyce potential, in agreement with \cite{Barbieri:2018swu}:
\bea
F&=&\frac{1}{2\pi\I} \sum_{\gamma\in\Gamma_{\!\star} } \sigma_\gamma\bOm(\gamma)\, e^{2\pi\I \theta_{\gamma}} = \frac{1}{2\pi\I} \sum_{\gamma\in\Gamma_{\!\star} }
\Omega(\gamma)\, \Li_2\(e^{2\pi\I \vth_{\gamma}}\)
\nn\\
&=&-\pi\I\sum_{\gamma\in\Gamma_{\!+} } \Omega(\gamma)\, B_2([\vth_\gamma]),
\eea
Note that our approach allows to get these potentials without evaluating explicitly the integrals
representing the holomorphic Fourier modes $\cX_\gamma$ which,
as we shall see now, are given by much more complicated expressions.

The main simplification of the uncoupled BPS structure is that the TBA-like {\it equations} \eqref{TBAeq}
become explicit integral {\it expressions} for the holomorphic Fourier modes $\cX_\gamma$.
As explained in \cite{Gaiotto:2014bza,Barbieri:2018swu} and recalled in Appendix \ref{sec-Xfinite},
the integral can be computed using Binet's second  formula for the Gamma function, leading to the generalized
Gamma function $\Lambda(z,\eta)$ defined in \eqref{defLambda}.
More precisely, one obtains
\be
\label{Xuncoupled}
\cX_\gamma(t) = e^{2\pi\I(\vth_{\gamma} - Z_\gamma/t)} \prod_{\gamma'\in \Gamma \atop \Re(Z_{\gamma'}/t)>0}
\Lambda\( \frac{Z_{\gamma'}}{t},1-[\vth_{\gamma'}] \)^{\Omega(\gamma') \langle \gamma',\gamma\rangle}.
\ee
For $\vth_\gamma=0$, \eqref{Xuncoupled} reduces to  \cite[Thm 3.2]{Bridgeland:2016nqw}
(after flipping the sign of the symplectic product, see footnote \ref{foot-sign}).
For non-vanishing $\vth_\gamma$, the solution considered in \cite{Barbieri:2018swu,Bridgeland:2019fbi}
corresponds to the analytic continuation of the branch obtained for $\Re\vth_\gamma\in (0,1)$.

Finally, the $\tau$ function is obtained from \eqref{restau}.
Rewriting it in terms integer BPS indices,
combining the contributions from $\pm\gamma$ and using the integral representation \eqref{BinetUps2}
for the generalized Barnes function $\Upsilon(z,\eta)$ defined in \eqref{defUpsilon},
we arrive after some algebra (see Appendix \ref{sec_Apptau} for details) at
\be
\log\tau = \sum_{\gamma\in \Gamma \atop \Re(Z_{\gamma}/t)>0}
\Omega(\gamma)\,\log \Upsilon\(\frac{Z_\gamma}{t},1-[\vth_\gamma]\) .
\label{taufun-res2}
\ee
For $\vth_\gamma=0$, the result \eqref{taufun-res2} reduces to \cite[Thm 3.4]{Bridgeland:2016nqw}.
The fact that the $\tau$ function \eqref{taufun-res2} is consistent with the solution \eqref{Xuncoupled},
i.e. satisfies the differential equation \eqref{eqtau}, can easily be seen from the relation \eqref{dUps}
between the generalized Barnes and Gamma functions. From \eqref{Upsexp}, it follows that the $\tau$ function
has the following expansion at large $Z_\gamma/t$
\be
\label{tauexp}
\log\Upsilon(z,\eta) = \sum_{\gamma\in \Gamma \atop \Re(Z_{\gamma}/t)>0}\[
\frac{1}{2}\, B_2([\vth_\gamma]) \log \frac{t}{Z_\gamma}
+ \sum_{k= 3}^\infty  \frac{  B_k([\vth_\gamma])}{k(k-2)}\(\frac{t}{Z_\gamma}\)^{k-2}\],
\ee
which generalizes the result \cite[Eq.(19)]{Bridgeland:2016nqw}.

\section{Refined Joyce structure}
\label{sec_def}

In this last section, we consider the deformed version of the heavenly equation \eqref{heaveq}
and isomonodromy equation \eqref{isoeq} introduced in \cite{Bridgeland:2020zjh}, and explain
how formal solutions to these equations can be constructed in terms of solutions of deformed TBA equations
proposed in \cite{Alexandrov:2019rth}, using the refined DT invariants
$\Omega(\gamma,y)$ as input data. Throughout we set $y=e^{2\pi\I\alpha}$ where $\alpha$ is the deformation parameter,
and let $\Delta(\alpha)$ be a function of $\alpha$ which behaves as
$\Delta(\alpha)=4\pi\I\alpha+O(\alpha^2)$ in the unrefined limit $\alpha\to 0$.
Thanks to conformal invariance, the choice of $\Delta(\alpha)$ is irrelevant,
but for concreteness, one may choose either $\Delta=4\pi\I\alpha$ (following
\cite{Bridgeland:2020zjh}) or $\Delta=y-1/y$
(following \cite{Alexandrov:2019rth} where this choice is required by modularity).
Moreover, we set $\kappa(x):=(y^x-y^{-x})/\Delta$ and denote $\gamma_{ij}=\langle\gamma_i,\gamma_j\rangle$.

\subsection{Refined construction}

For any two functions $f,g$ on $\cZ=\IP^1\times \cM$, let us introduce the non-commutative Moyal star product
\be
f \star g =  f \exp\[ \frac{\alpha}{2\pi\I}\sum_{a,b}\omega_{ab}\,
\overleftarrow{\p}_{\!\theta^a}\overrightarrow{\p}_{\!\theta^b} \] g
\label{starproduct-alt}
\ee
and the rescaled Moyal bracket
\be
\{f,g\}_\star=\frac{1}{\Delta}\(f \star g-g\star f\).
\label{Mbracket}
\ee
In particular, the semi-flat Fourier modes \eqref{cXsf} satisfy
\be
\begin{split}
\cXsf_{\gamma_1}(t_1)\star\cXsf_{\gamma_2}(t_2) =&\,
y^{\gamma_{12}}\cXsf_{\gamma_1}(t_1)\cXsf_{\gamma_2}(t_2),
\\
\left\{ \cXsf_{\gamma_1}(t_1) , \cXsf_{\gamma_2}(t_2) \right\}_\star =&\,
\kappa(\gamma_{12}) \cXsf_{\gamma_1}(t_1)\cXsf_{\gamma_2}(t_2).
\end{split}
\label{starXX}
\ee

In \cite{Bridgeland:2020zjh} (see also \cite{Strachan:1992em,Takasaki:1992jf} for earlier work) it was proposed
that the Joyce structure has a natural refined version described by the deformed heavenly equation
\be
\frac{\p^2 \Wr}{\p z^b\p \theta^a}-\frac{\p^2 \Wr}{\p z^a\p \theta^b}
=\left\{ \frac{\p \Wr}{\p \theta^a}\, ,\,  \frac{\p \Wr}{\p \theta^b}\right\}_\star.
\label{heaveq-ref}
\ee
Furthermore, introducing the refined Joyce potential $\Fr$ through the same relation \eqref{FW} as before,
it was shown that its Fourier modes then satisfy the deformed isomonodromy equation
\be
\de \Fr_\gamma=\sum_{\gamma_1+\gamma_2=\gamma}
 \kappa(\gamma_{12}) \, \Fr_{\gamma_1}\Fr_{\gamma_2}\, \de\log Z_{\gamma_2}.
\label{isoeq-ref}
\ee
We note that this equation can be rewritten using the same Moyal bracket as in \eqref{heaveq-ref}
so that it appears as a straightforward deformation of the version of the isomonodromy equation given in \eqref{isomondr}
\be
\frac{\p\Fr}{\p z^a}=\left\{ \Fr, \p_{\theta^a}\Wr\right\}_\star.
\label{isoref}
\ee

It turns out that it is possible to construct solutions to these equations in a way closely similar to
the unrefined construction presented above.
Let us consider the TBA-like system of equations, similar to \cite[Eq.(3.43)]{Alexandrov:2019rth},
\be
\cXr_{\gamma}(t) =\cXsf_{\gamma}(t) \star\(1
+
\sum_{\gamma'\in \Gamma_\star}
\frac{\bOm(\gamma',y) }{2\pi\I \Delta}  \, \int_{\ell_{\gamma'}} \frac{\de t'}{ t' }
\frac{t}{t'-t}\ \cXr_{\gamma'}(t')\).
\label{inteqH-star}
\ee
This system of equations will play the same role as \eqref{TBAeq1}, but it is worth stressing that
it does {\it not} reduce to \eqref{TBAeq1} in the unrefined limit $y\to1$.
In fact, the system \eqref{inteqH-star} is singular in this limit, which indicates
that the refined Fourier modes $\cXr_{\gamma}$ do not have a smooth limit as $y\to 1$.
Nevertheless, it is still possible to produce a solution for $\cXr_{\gamma}$
as a formal series in refined, rational DT invariants  by the usual iterative  procedure.

Assuming that a solution of \eqref{inteqH-star} has been found, let us define
\bea
\Wr &=&\frac{1}{2\pi\I}\, \sum_{\gamma\in \Gamma_\star} \bOm(\gamma,y) \int_{\ell_\gamma}
\frac{\de t}{t^2}\, \cXr_\gamma,
\label{Wref}
\\
\Fr &=& \frac{1}{2\pi\I}\, \sum_{a=1}^m z^a
\sum_{\gamma\in \Gamma_\star} \bOm(\gamma,y) \int_{\ell_\gamma}
\frac{\de t}{t^2}\, \p_{\theta^a}\cXr_{\gamma}.
\label{FWr}
\eea
These definitions are analogous to \eqref{Wfull-short} and \eqref{Ffull-short}, but in contrast to the unrefined case,
$\Wr$ involves no quadratic term.
It is easier to check that $\Wr$  and $\Fr$ do satisfy the relations \eqref{FW} and \eqref{defW}, and we claim
that they also satisfy the deformed heavenly and isomonodromy equations, \eqref{heaveq-ref} and \eqref{isoref}.
The proofs of these statements can be found in Appendix \ref{sec_Appdef}.

While these results allow to construct formal solutions of the deformed equations,
we do not know at this point if they converge to actual solutions, nor if they carry any physical or mathematical meaning.
They also do not seem to simplify in the mutually local case.
It would be interesting to make contact with the construction in \cite{Barbieri:2019yya} in the double $A_1$ case.

\subsection{Unrefined limit}

While the solutions to \eqref{inteqH-star} do not seem to bear any simple relation to the solutions
of the undeformed system \eqref{TBAeq1},  we claim that the refined potentials
\eqref{Wref} and \eqref{FWr} do reduce in the unrefined limit to \eqref{Wfull} and \eqref{Ffull}, respectively.
To show this, we start from the formal series obtained by iterating \eqref{inteqH-star} and plugging it into \eqref{Wref},
\be
\Wr=\frac{1}{2\pi\I}\,\sum_{n=1}^\infty\IS_1\cXsf_1 \prod_{i=2}^n \( \IS_i \Kr_{i-1,i} \cXsf_i \) y^{\sum_{k<l}\gamma_{kl}},
\label{expWr}
\ee
where we used the notations similar to \eqref{short}
\be
\Kr_{ij}=\frac{1}{2\pi\I\Delta}\,\frac{t_i t_j}{t_j-t_i}\, ,
\qquad
\IS_i=\sum_{\gamma_i\in\Gamma_{\!\star} }
\bOm(\gamma_i,y)\int_{\ell_i}\frac{\de t_i}{t_i^2}\, .
\ee
Note that in contrast to the undeformed case, the iterative solution does not involve a sum over trees (cf. \eqref{iterW})
and the effect of the Moyal product is completely captured by the last $y$-dependent factor.

To extract the unrefined limit of the resulting series, we propose the following two conjectures:
 \begin{conj}
\bea
a)&& \ \Sym\left\{y^{\sum_{k<l}\gamma_{kl}}\prod_{i=1}^{n-1}\Kr_{i,i+1} \right\}
=\frac{\Delta^{n-1}}{n}\, \Sym\left\{
\prod_{i=1}^{n-1} \kappa\( \langle \Gamma_i,\gamma_{i+1} \rangle\)
\prod_{i=1}^{n-1}\Kr_{i,i+1} \right\},
\\
b)&&  \
\frac{1}{n}\Sym\left\{
\prod_{i=1}^{n-1} \Bigl(\langle\Gamma_i,\gamma_{i+1}\rangle \,
\frac{t_i t_{i+1}}{t_{i+1}-t_i} \Bigr)\right\}
=\frac{1}{2^{n-1}}\, \Sym\left\{
\[  \sum_{\cT\in \mathbb{T}_n^\ell}\,
 \prod\limits_{e: i\to j}  \langle \gamma_i, \gamma_j \rangle \]
\prod_{i=1}^{n-1} \frac{t_i t_{i+1}}{t_{i+1}-t_i} \right\},
\nn\\&&
\label{limWn}
\eea
where $\Gamma_k=\sum_{i=1}^k \gamma_i$ and $\mathbb{T}_n^\ell$ is the set of trees defined below
\eqref{FJoyce}.
\end{conj}
\noindent
We have checked these two conjectures using computer experiments
with random values of $t_i$ and $\gamma_{ij}$.
The first property allows to rewrite the integration kernel in the series \eqref{expWr} in the form which has
a well defined limit $y\to 1$ because in this limit $\kappa(x)\to x$.
The second property ensures the exact match with the expansion of
the unrefined Pleba\'nski potential $W$, which is obtained either by combining \eqref{iterW} with Conjecture \ref{conj-KT},
or simply by dividing each Fourier modes \eqref{FJoyce} by $Z_\gamma$.
The same reasoning shows that $\Fr$ also reduces as $y\to 1$ to its unrefined counterpart.

\section*{Acknowledgements}
We are indebted to Tom Bridgeland for his nice set of lectures at the work of the Simons collaboration
on Special Holonomy in Geometry, Analysis and Physics (Jan 11-13, 2021),
which stimulated our interest in this topic, and for subsequent
correspondence. We are also grateful to Andy Neitzke for discussions.

\appendix

\section{Proofs}
\label{ap-proof}

In this appendix we present several proofs which have been omitted in the main text.

\subsection{Heavenly equation}
\label{ap-proofHE}

Let us show that the function $W$ \eqref{Wfull-short} satisfies the heavenly equation \eqref{heaveq1}.
As in the proof of the section condition,
we do this with help of an infinite series generated by iterating the relation \eqref{derX-sh}.
Namely, using \eqref{p2W}, we get
\bea
&&
\frac{\p^2 W}{\p z^a\p \theta^b}-\frac{\p^2 W}{\p z^b\p \theta^a}
-\frac{1}{(2\pi)^2}\sum_{c,d}\omega^{cd} \, \frac{\p^2 W}{\p \theta^a\p\theta^c}\, \frac{\p^2 W}{\p \theta^b\p\theta^d}
\\
&=& -\IS_1\, \frac{1}{t_1}\(q_{1,a}
\p_{\theta^b} \cX_1 -q_{1,b}\p_{\theta^a}\cX_1\)
-\frac{1}{(2\pi)^2}\, \IS_1 \p_{\theta^a}\cX_1 \IS_2
\p_{\theta^b} \cX_2 \gamma_{12}
\nn\\
&=& 2\pi\I\, \IS_1 \frac{1}{t_1}\, \cX_1 \sum_{n=1}^\infty \prod_{i=2}^{n} \( \IS_i K_{i-1,i} \cX_i \)
\(q_{1,b} q_{n,a}-q_{1,a} q_{n,b}\)
\nn\\
&&
+ \IS_1 \cX_1\sum_{n'=1}^\infty \prod_{i=2}^{n'} \( \IS_i K_{i-1,i} \cX_i \)q_{n',a}
\IS_{n'+1}\cX_{n'+1} \langle\gamma_1,\gamma_{n'+1} \rangle
\sum_{n''=1}^\infty \prod_{i=n'+2}^{n'+n''} \( \IS_i K_{i-1,i} \cX_i \)q_{n'+n'',b}
\nn\\
&=& 2\pi\I \IS_1 \cX_1 q_{1,a}\sum_{n=1}^\infty \prod_{i=2}^n \( \IS_i K_{i-1,i} \cX_i \)q_{n,b}
\(\frac{1}{t_n}-\frac{1}{t_1}+\sum_{i=1}^{n-1}\frac{t_{i+1}-t_i}{t_i t_{i+1}}\)=0,
\eea
where in the last line we inverted the labeling of first $n'$ variables and set $n=n'+n''$.
$\Box$

\subsection{Isomonodromy equation}
\label{ap-proofISO}

To prove that the isomonodromy equation  \eqref{isoeq},
we first rewrite it in terms of the potentials $F$ and $W$ which can be achieved by multiplying by $e^{2\pi\I \theta_\gamma}$
and summing over the lattice. This leads to the following equation
\be
\de F=\frac{1}{(2\pi\I)^2}\sum_{a, b,c}\omega^{ab}\p_{\theta^a}F \p_{\theta^b}\p_{\theta^c}W\, \de z^a,
\label{isomondr}
\ee
where the differential is supposed to act only on the variables $z^a$.
Substituting the expression \eqref{Ffull-short} for the Joyce potential and using \eqref{p2W} and the equation \eqref{TBAeq-sh},
one arrives at the condition
\be
\IS_1 \[\cX_1\(1-2\pi\I\,\frac{Z_{\gamma_1}}{t_1}\)\de Z_{\gamma_1} +Z_{\gamma_1}\cX_1\IS_2 K_{12}\de \cX_2 \]
=\frac{1}{(2\pi\I)^2} \sum_{a,b}\omega^{ab}\IS_1Z_{\gamma_1}
\p_{\theta^a}\cX_1\IS_2 \p_{\theta^b} \cX_{\gamma_2}\de Z_{\gamma_2}.
\label{isoeq-proof1}
\ee
Then note that the first term on the l.h.s. can be rewritten as
\be
\begin{split}
&
\IS_1 \cX_1\(1-2\pi\I\,\frac{Z_{\gamma_1}}{t_1}\)\de Z_{\gamma_1}
=\IS_1 \(\(1-t_1\p_{t_1}\)\cX_1+\cX_1 \,t_1\p_{t_1}\IS_2 K_{12} \cX_2\)\de Z_{\gamma_1}
\\
=&\, \IS_1 \cX_1\,\de Z_{\gamma_1}\IS_2 K_{12} t_2\p_{t_2}\cX_2.
\end{split}
\ee
Substituting this result into \eqref{isoeq-proof1} and bringing all terms to the l.h.s., the resulting expression becomes
\bea
&&
\IS_1\cX_1\,\de Z_{\gamma_1}\,\IS_2 K_{12}t_2\p_{t_2}\cX_2+\IS Z_{\gamma_1}\cX_1\IS_2 K_{12}\de\cX_2
-\frac{1}{(2\pi\I)^2} \sum_{i,j}\omega^{ab}\IS_1 Z_{\gamma_1}
\p_{\theta^a} \cX_1\IS_2\p_{\theta^b} \cX_{\gamma_2}\de Z_{\gamma_2}
\nn\\
&=&2\pi\I\[\sum_{n=1}^\infty\IS_1\cX_1\,\de Z_{\gamma_1} \prod_{i=2}^n \(\IS_i K_{i-1,i} \cX_i \) \frac{Z_{\gamma_n}}{t_n}
-\sum_{n=1}^\infty\IS_1\cX_1\,Z_{\gamma_1} \prod_{i=2}^n \(\IS_i K_{i-1,i} \cX_i \) \frac{\de  Z_{\gamma_n}}{t_n}\]
\nn\\
&&
-\IS_1 Z_{\gamma_1}\cX_1\sum_{n'=1}^\infty \prod_{i=2}^{n'} \( \IS_i K_{i-1,i} \cX_i \)
\IS_{n'+1}\cX_{n'+1}\de Z_{\gamma_{n'+1}} \sum_{n''=1}^\infty \prod_{i=n'+2}^{n'+n''}
\( \IS_i K_{i-1,i} \cX_i \)\langle\gamma_{n'},\gamma_{n'+n''}\rangle
\nn\\
&=&2\pi\I\sum_{n=1}^\infty\IS_1 Z_{\gamma_1}\cX_1\prod_{i=2}^n \( \IS_i K_{i-1,i} \cX_i\)
\(\frac{1}{t_1}-\frac{1}{t_n}-\sum_{i=1}^{n-1}\frac{t_{i+1}-t_i}{t_it_{i+1}}\)\de  Z_{\gamma_n}=0,
\eea
where in the last line we inverted the labeling of the variables labeled from $n'+1$ to $n'+n''$ on the previous line.
$\Box$

\subsection{Conjecture 1 at low rank}
\label{ap-KT}

Here we prove Conjecture \ref{conj-KT} presented in \S\ref{sec_asym} in the simplest cases $n=3$ and $n=4$.
Let us denote $\gamma_{ij}=\langle\gamma_i,\gamma_j\rangle$ and $k_{ij}=\frac{t_i t_j}{t_j-t_i}$ so that $K_{ij}=\gamma_{ij} k_{ij}/(2\pi\I)$.
The functions $k_{ij}$ satisfy
\be
k_{ij}k_{jk}-k_{ik}k_{jk}-k_{ij}k_{ik}=0.
\label{identker}
\ee
The proof can be reduced essentially to a repetitive use of this identity.

\subsubsection{$n=3$}

In this case we need to show that
\be
\begin{split}
\Sym\Bigl[ \gamma_{12}\gamma_{23} \,k_{12} k_{23}\Bigr]
=&\, \hf\Sym\Bigl[ \( \gamma_{12}\gamma_{23}+ \gamma_{12}\gamma_{13}+ \gamma_{13}\gamma_{23}\) k_{12} k_{23}\Bigr]
\\
=&\, \hf\Sym\Bigl[ \gamma_{12}\gamma_{23}\(k_{12} k_{23}+k_{12}k_{13}+k_{13}k_{23}\)\Bigr].
\end{split}
\label{conjprove3}
\ee
This is a direct consequence of \eqref{identker} because it ensures that the expression in the round brackets on the r.h.s.
is equal to $2k_{12}k_{23}$.

\subsubsection{$n=4$}

In this case we need to prove two relations
\be
\begin{split}
&\,\Sym\Bigl[ \gamma_{12}\gamma_{23}\gamma_{34} \,k_{12} k_{23}k_{34}\Bigr]
= \frac14\Sym\Bigl[ \( \gamma_{12}\gamma_{23}\gamma_{34}
+ \gamma_{12}\gamma_{13}\gamma_{34}+\gamma_{12}\gamma_{24}\gamma_{34}+\gamma_{12}\gamma_{14}\gamma_{34}
\right.
\\
&\, \qquad
+ \gamma_{13}\gamma_{23}\gamma_{24}+\gamma_{14}\gamma_{23}\gamma_{24}+\gamma_{12}\gamma_{23}\gamma_{14}
+\gamma_{13}\gamma_{24}\gamma_{34}+\gamma_{12}\gamma_{13}\gamma_{24}+\gamma_{14}\gamma_{23}\gamma_{34}
\\
&\, \left.\qquad
+\gamma_{13}\gamma_{23}\gamma_{14}+\gamma_{13}\gamma_{14}\gamma_{24}
\)k_{12} k_{23}k_{34}\Bigr],
\end{split}
\ee
\be
\frac13 \Sym\Bigl[ \gamma_{12}\gamma_{13}\gamma_{14} \,k_{12} k_{13}k_{14}\Bigr]
= \frac14\Sym\Bigl[ \( \gamma_{12}\gamma_{13}\gamma_{14}+\gamma_{12}\gamma_{23}\gamma_{24}+
\gamma_{13}\gamma_{23}\gamma_{34}+\gamma_{14}\gamma_{24}\gamma_{34}\)k_{12} k_{23}k_{34}\Bigr].
\label{rel2-kg}
\ee
Proceeding as in \eqref{conjprove3}, in each relation we rewrite all terms so that
they have the same product of $\gamma_{ij}$ factors.

The first relation is then equivalent to
\be
\begin{split}
&\,
k_{12}k_{23}k_{34}+k_{12}k_{13}k_{34}+k_{12}k_{24}k_{34}+k_{12}k_{14}k_{34}+k_{13}k_{23}k_{24}+k_{13}k_{24}k_{34}
+k_{13}k_{23}k_{14}
\\
&\,
+k_{14}k_{23}k_{34}
+k_{12}k_{23}k_{14}+k_{14}k_{23}k_{24}
+k_{12}k_{13}k_{24}-k_{13}k_{14}k_{24}=4 k_{12}k_{23}k_{34}.
\end{split}
\label{rel-k1}
\ee
The last 8 terms on the l.h.s. combined pairwise using the identity \eqref{identker} give
\be
2k_{13}k_{23}k_{34}+k_{12}k_{23}k_{24}+k_{12}k_{13}k_{14}.
\ee
The last two terms can in turn be combined with the third and fourth terms in \eqref{rel-k1} so that
the full l.h.s. becomes
\be
2\(k_{12}k_{23}k_{34}+k_{12}k_{13}k_{34}+k_{13}k_{23}k_{34}\).
\ee
Due to \eqref{identker}, the sum of the last two terms here is equal to the first one
so that one indeed obtains $4 k_{12}k_{23}k_{34}$ as required.

The second relation \eqref{rel2-kg} follows from
\be
\begin{split}
&\,
k_{12}k_{23}k_{34}+k_{12}k_{14}k_{23}-k_{13}k_{14}k_{23}+k_{12}k_{13}k_{24}-k_{13}k_{14}k_{24}
-k_{13}k_{23}k_{24}-k_{14}k_{23}k_{24}
\\
&\,
+k_{12}k_{13}k_{34}-k_{12}k_{14}k_{34}+k_{14}k_{23}k_{34}
-k_{12}k_{24}k_{34}-k_{13}k_{24}k_{34}=4k_{12}k_{13}k_{14}.
\end{split}
\ee
To prove it, we note that combining the following pairs of terms on the l.h.s., (2,3), (4,5), (8,9), (1,11), (6,10), (5,12),
one finds
\be
3k_{12}k_{13}k_{14}+k_{12}k_{23}k_{24}+k_{14}k_{24}k_{34}-k_{13}k_{23}k_{34}.
\ee
This further can be rewritten as
\be
\begin{split}
&\,
3k_{12}k_{13}k_{14}+(k_{12}k_{13}+k_{13}k_{23})k_{24}+k_{24}(k_{13}k_{34}-k_{13}k_{14})
-k_{13}(k_{24}k_{34}+k_{23}k_{24})
\\
=&\, 3k_{12}k_{13}k_{14}+k_{12}k_{13}k_{24}-k_{24}k_{13}k_{14}
\end{split}
\ee
Since the sum of the last two terms is combined into the first one, one indeed obtains $4k_{12}k_{13}k_{14}$ as required.

\subsection{Iterated integrals}
\label{sec_JJt}

Here we prove that the functions $J(z_1,\dots, z_n)$ defined in \eqref{Jours}
satisfies the axioms a)-d) spelled out in \eqref{axiomsJ}.

\noindent a) The homogeneity and holomorphicity are manifest from the representation on the second line of \eqref{defJtn}.

\noindent b)
Let us denote $z_{k,l}=\sum_{i=k}^l z_i$ and $z=z_{1,n}$.
Then bringing all terms in the recursion relation \eqref{dFJoyce} to one side, one finds
\be
\label{dFJoyce-our}
\begin{split}
&\,
\de J_n(z_1,\dots, z_n)
- \sum_{k=1}^{n-1}\, J_k(z_1,\dots z_k)\,J_{n-k}(z_{k+1},\dots z_n)
\left[ \frac{\de z_{k+1,n}}{z_{k+1,n}} -\frac{\de z_{1,k}}{z_{1,k}}\right]
\\
=&\,
\frac{1}{(2\pi\I)^{n-1}}\prod_{i=1}^n\( \int_{\ell_{\gamma_i}}\frac{\de t_i}{t_i^2}\, e^{-2\pi\I z_i/t_i}\)
\prod_{i=1}^{n-1} \frac{t_i t_{i+1}}{t_{i+1}-t_i}
\[\de z- 2\pi\I z\sum_{k=1}^{n}\frac{\de z_k}{t_k}
\right.
\\
&\,
\left. \qquad
-2\pi\I \sum_{k=1}^{n-1}\(\frac{1}{t_k}-\frac{1}{t_{k+1}}\)\(z_{1,k}\de z_{k+1,n} -z_{k+1,n}\de z_{1,k}\)\].
\end{split}
\ee
Since the  last line sums up to
$-2\pi\I\sum_{k=1}^{n}\frac{1}{t_k}\( z_k\de z-z\de z_k\)$,
the total expression reduces to
\be
\frac{\de z}{(2\pi\I)^{n-1}}\prod_{i=1}^n\( \int_{\ell_{\gamma_i}}\frac{\de t_i}{t_i^2}\, e^{-2\pi\I z_i/t_i}\)
\prod_{i=1}^{n-1} \frac{t_i t_{i+1}}{t_{i+1}-t_i}
\(1- 2\pi\I \sum_{k=1}^{n}\frac{z_k}{t_k} \).
\ee
Changing variables to $t_i=\I z_i/(s u_i)$ with $s, u_i\in \IR^+$ and $\sum_{i=1}^n u_i=1$,
the integral factorizes exactly as in \eqref{defJtn} into
\be
- \frac{\I \de z}{(2\pi\I)^{n-1} \prod_{i=1}^n z_i}
\int_0^\infty \de s\(1-2\pi s\) e^{-2\pi s}\,
\int_{0\leq u_1\leq 1\atop \sum_{i=1}^n u_1=1} \de u_1\cdots \de u_{n-1}\,
\prod_{i=1}^{n-1} \frac{z_i z_{i+1}}{z_{i+1} u_i-z_i u_{i+1}}\, .
\ee
Since the integral over $s$ vanishes, this proves the recursion relation \eqref{dFJoyce}.

\noindent c) The vanishing property \eqref{vanishsym} is a consequence of the identity
\be
\label{sumsym}
\sum_{s\in \Sigma_n}\prod_{i=1}^{n-1}\frac{t_{s(i)} t_{s(i+1)}}{t_{s(i+1)}-t_{s(i)}}=0.
\ee
To show that this holds, let $u_i=1/t_i$, and consider
\be
S_n(u_1,\dots, u_n) = \sum_{s\in \Sigma_n}\prod_{i=1}^{n-1}
\frac{1}{u_{s(i)}- u_{s(i+1)}}\, .
\ee
We can show that $S_n=0$ inductively on $n$. For $n=1$ and $n=2$, this is clear.
For $n>2$, the residue of $S_n$ at $u_{n-1}=u_n$ is equal to $S_{n-2}(u_1,\dots, u_{n-2})$,
which vanishes by the induction hypothesis. By symmetry, the same holds for the residue at $u_i=u_j$.
Since $S_n$ is a rational function with no poles, it must be a constant.
And since it vanishes at large $u_i$, it vanishes identically.

\noindent d) The growth condition is evident from the representation \eqref{integral-w}
since  singularities can only arise when the integration region $R_n(z_1,\dots, z_n)$ touches
one of the divisors $w_i=0$, leading at most to logarithmic singularities. The behavior
as $z_k\to 0$ is  easily read off from \eqref{defJtn}:
\bea
J_n(z_1,\dots, z_n) &\stackrel{z_1\to 0}{\sim} & -\frac{1}{(2\pi\I)^2}\,  J_{n-1}(z_2,z_3,\dots, z_n)\log(z_1),
\nn\\
J_n(z_1,\dots, z_n) &\stackrel{z_n\to 0}{\sim} & \frac{1}{(2\pi\I)^2}\,  J_{n-1}(z_1,\dots, z_{n-1})\log(z_n),
\eea
where we used the explicit form of $J_2$ \eqref{J2},
whereas $J_n(z_1,\dots, z_n)$ stays finite as $z_k\to 0$ for $1<k<n$, as the logarithmic
divergences at $u_k=z_k u_{k\pm 1}/z_{k\pm 1}$ cancel. $\Box$

\subsection{Refined potentials}
\label{sec_Appdef}

Here we prove that the refined versions of the Pleba\'nski the Joyce potentials, \eqref{Wref} and \eqref{FWr},
satisfy the deformed heavenly equation \eqref{heaveq-ref} and isomonodromy equation \eqref{isoref}, respectively.
As in the unrefined case, the proof relies on an asymptotic expansion. But the difference is that now this will be
the expansion in powers of $\cXsf_\gamma$.

Thus, the starting point to establish \eqref{heaveq-ref} is the formal series \eqref{expWr}
which implies
\bea
\frac{\p^2 \Wr}{\p z^a\p \theta^b}&=&
-2\pi\I\sum_{n=1}^\infty\IS_1\cXsf_1 \prod_{i=2}^n \( \IS_i \Kr_{i-1,i} \cXsf_i \)
y^{\sum_{k<l}\gamma_{kl}}\sum_{k=1}^n\frac{q_{k,a}}{t_k}\sum_{\ell=1}^n q_{\ell,b}\, ,
\\
\p_{\theta^a} \Wr \star\p_{\theta^b} \Wr&=&
\sum_{n'=1}^\infty \IS_1 \cXsf_1 \prod_{i=2}^{n'} \( \IS_i \Kr_{i-1,i} \cXsf_i \)
y^{\sum_{l>k=1}^{n'} \gamma_{kl}}
\sum_{k=1}^{n'} q_{k,a}
\\
&\times&
\sum_{n''=1}^\infty\IS_{n'+1}\cXsf_{n'+1} \prod_{i=n'+2}^{n'+n''} \( \IS_i \Kr_{i-1,i} \cXsf_i \)
y^{\sum_{l>k=n'+1}^{n'+n''} \gamma_{kl}}\sum_{\ell=n'+1}^{n'+n''} q_{\ell,b}
\, y^{\sum_{k=l}^{n'}\sum_{l=n'+1}^{n'+n''}\gamma_{kl}}
\nn\\
&=&\sum_{n=1}^\infty \IS_1 \cXsf_1 \prod_{i=2}^n \( \IS_i \Kr_{i-1,i} \cXsf_i \)
y^{\sum_{k<l}\gamma_{kl}}
\sum_{m=1}^{n-1} \frac{\sum_{k=1}^m q_{k,a}
\sum_{\ell=m+1}^n q_{\ell,b}}{\Kr_{m,m+1}}
\nn
\eea
where in the last line we, as usual, set $n=n'+n''$.
Substituting these results into the deformed heavenly equation, one obtains
\be
\begin{split}
&
\frac{\p^2 \Wr}{\p z^b\p \theta^a}-\frac{\p^2 \Wr}{\p z^a\p \theta^b}
-\left\{ \frac{\p \Wr}{\p \theta^a}\, ,\,  \frac{\p \Wr}{\p \theta^b}\right\}_\star
\\
=&\, 2\pi\I \sum_{n=1}^\infty \IS_1 \cXsf_1 \prod_{i=2}^n \( \IS_i \Kr_{i-1,i} \cXsf_i \)
y^{\sum_{k<l}\gamma_{kl}} S_{ab},
\end{split}
\ee
where
\be
S_{ab}=\sum_{k,l=1}^N\frac{1}{t_k}
\(q_{k,a}q_{l,b}-q_{k,b}q_{l,a}\)
-\sum_{m=1}^{n-1} \(\frac{1}{t_m}-\frac{1}{t_{m+1}}\)\sum_{k=1}^m\sum_{l=m+1}^n
\(q_{k,a}q_{l,b}-q_{k,b}q_{l,a}\).
\ee
It is straightforward to verify that $S_{ab}$ actually vanishes. Indeed, denoting $Q_{m,a}=\sum_{k=1}^m q_{k,a}$, we have
\bea
S_{ab}&=& \sum_{m=1}^n\frac{1}{t_m}\(q_{m,a}Q_{n,b}-q_{m,b}Q_{n,a}\)
-\sum_{m=1}^{n-1} \(\frac{1}{t_m}-\frac{1}{t_{m+1}}\)
\(Q_{m,a}Q_{n,b}-Q_{m,b}Q_{n,a}\)
\nn\\
&=& \frac{1}{t_n}\(q_{n,a}Q_{n,b}-q_{n,b}Q_{n,a}\)
+\frac{1}{t_n}\(Q_{n-1,a}Q_{n,b}-Q_{n-1,b}Q_{n,a}\)=0.
\eea
Thus, the deformed heavenly equation \eqref{heaveq-ref} holds.

In a similar way, to establish \eqref{isoref}, we compute
\bea
&&
\de \Fr-\sum_{a}\{ \Fr, \p_{\theta^a}\Wr\}_\star\,\de z^a
\nn\\
&=&
\sum_{n=1}^\infty\IS_1\cXsf_1 \prod_{i=2}^n \( \IS_i \Kr_{i-1,i} \cXsf_i \)
\(\de \zeta_n-2\pi\I \zeta_n\sum_{k=1}^n \frac{\de Z_k}{t_k}\)
y^{\sum_{k<l}\gamma_{kl}}
\nn\\
&&-\sum_{n'=1}^\infty \IS_1 \cXsf_1 \prod_{i=2}^{n'} \( \IS_i \Kr_{i-1,i} \cXsf_i \)
y^{\sum_{l>k=1}^{n'}\gamma_{kl}}
\sum_{k=1}^{n'} Z_k
\label{proof-isoref}\\
&&\times
\sum_{n''=1}^\infty\IS_{n'+1}\cXsf_{n'+1} \prod_{i=n'+2}^{n'+n''} \( \IS_i \Kr_{i-1,i} \cXsf_i \)y^{\sum_{l>k=n'+1}^{n'+n''}\gamma_{kl}}
\, \kappa\(\sum_{k=l}^{n'}\sum_{l=n'+1}^{n'+n''}\gamma_{kl}\)\sum_{k=n'+1}^{n'+n''}\de Z_k.
\nn
\eea
where $Z_k=Z_{\gamma_k}$, $\zeta_m=\sum_{k=1}^m Z_k$ and the function $\kappa(x)$ was defined at the beginning of \S\ref{sec_def}.
Further noting that
\bea
\IS_1\cXsf_1 \prod_{i=2}^n \( \IS_i \Kr_{i-1,i} \cXsf_i \) &=&
\IS_1 t_1\p_{t_1}\[\cXsf_1 \prod_{i=2}^n \( \IS_i \Kr_{i-1,i} \cXsf_i \)\]
\nn\\
&=& \IS_1\cXsf_1\IS_2 \Kr_{12}\( 2\pi\I\, \frac{Z_1}{t_1}+t_2\p_{t_2}\)\[\cXsf_2\prod_{i=3}^n \( \IS_i \Kr_{i-1,i} \cXsf_i \)\]
\nn\\
&=&2\pi\I \IS_1\cXsf_1 \prod_{i=2}^n \( \IS_n \Kr_{i-1,i} \cXsf_i \)\sum_{k=1}^n \frac{Z_k}{t_k}\, ,
\eea
we see that \eqref{proof-isoref} can be rewritten as
\be
2\pi\I \sum_{n=1}^\infty \IS_1 \cXsf_1 \prod_{i=2}^n \( \IS_i \Kr_{i-1,i} \cXsf_i \)
y^{\sum_{k<l}\gamma_{kl}} S'
\ee
where
\bea
S'&=&
\sum_{k=1}^n \frac{1}{t_k}\(Z_k\de \zeta_n -\zeta_n\de Z_k\)
-\sum_{m=1}^{n-1} \(\frac{1}{t_m}-\frac{1}{t_{m+1}}\)\(\zeta_m\de \zeta_n-\zeta_n\de \zeta_m\)
\nn\\
&=& \frac{1}{t_n}\(Z_n\de \zeta_n -\zeta_n\de Z_n\)
+\frac{1}{t_n}\(\zeta_{n-1}\de \zeta_n-\zeta_n \de \zeta_{n-1}\)=0.
\eea
This proves the deformed isomonodromy equation \eqref{isoref}.

\section{Special functions and useful identities}
\label{ap-fun}

In this appendix we collect the definitions and some properties of several special functions
as well as various useful identities used in \S\ref{sec-local}.

\subsection{Miscellaneous}

\subsubsection*{Bernoulli polynomials}

The Bernoulli polynomials have the following generating function
\be
\sum_{m=0}^\infty  \frac{x^m}{m!}\, B_m(\xi)= \frac{x\, e^{\xi x}}{e^x-1}.
\label{genBn}
\ee
The first few polynomials are given by
\be
\begin{split}
B_0(x)=&\, 1,
\\
B_1(x)=&\, x-\hf\, ,
\\
B_2(x)=&\, x^2-x+\frac16\, ,
\\
B_3(x)=&\, x^3-\frac32\, x^2+\frac12\,  x.
\end{split}
\label{Berpol}
\ee
At $x=0$ they reduce to Bernoulli numbers $B_n$ and have the following symmetry property
\be
B_n(1-x)=(-1)^n B_n(x).
\label{symBn}
\ee
Importantly, the Bernoulli polynomials arise in the inversion formula for polylogarithms, namely,
\be
\Li_n(e^{2\pi\I x})+(-1)^n\Li_n(e^{-2\pi\I x})=-\frac{(2\pi\I)^n}{n!}\, B_n([x]),
\label{Li-ident}
\ee
where
\be
[x]=
\begin{cases}
x-\fl{\Re x},  & \quad \mbox{if } \Im x\ge 0,
\\
x+\fl{-\Re x}+1,  & \quad \mbox{if } \Im x< 0.
\end{cases}
\label{defbr}
\ee
Note that for $\Im x\ne 0$ or $x\notin \IZ$, the bracket satisfies $[-x]=1-[x]$.

\subsubsection*{Useful identities}

For $\Re a,\Re b>0$, one has
\begin{subequations}
\bea
\int_0^\infty \frac{\de x}{x}\[e^{-ax} -e^{-bx}\]&=& -\log \frac{a}{b}\, ,
\label{int-expab1}
\\
\int_0^\infty \frac{\de x}{x^2}\[e^{-ax} -(1-(a-b)x)\,e^{-bx}\]&=& b-a+a\log \frac{a}{b}\, ,
\label{int-expab2}
\\
\int_0^\infty \frac{\de x}{x^3}\[e^{-ax} -\(1-(a-b)x+\hf\,(a-b)^2 x^2\)\,e^{-bx}\]&=&
\frac{3a^2+b^2}{4}-ab-\frac{a^2}{2}\,\log \frac{a}{b}\qquad
\label{int-expab3}
\eea
\label{int-expab}
\end{subequations}

\subsection{Integral representations of Gamma and Barnes functions}

\subsubsection*{Gamma function}

The logarithm of the Gamma function has two integral representations due to Binet
which are both valid for $\Re z>0$ \cite[p248-250]{whittaker2020course}:
\begin{itemize}
\item
first Binet formula
\be
\label{Binet1}
\log\Gamma(z) = \left(z-\frac12\right) \log z - z + \frac12\,\log(2\pi) +
\int_0^\infty \frac{\de s}{s} \(\hf-\frac{1}{s}+\frac{1}{e^s-1} \) e^{-z s},
\ee
\item
second Binet formula
\be
\log\Gamma(z)
= \(z-\frac12\) \log z - z + \frac12\,\log(2\pi) - \frac{z}{\pi} \int_0^{\infty} \frac{\de s}{s^2+z^2}\,
\log\( 1- e^{-2\pi s} \).
\label{Binet2}
\ee
\end{itemize}
The first formula has the following useful generalization due to Hermite,
valid for $\Re z>0$ and $\Re (z+\eta)>0$ \cite{nemes2013generalization}
\be
\log\Gamma(z+\eta) = \left(z+\eta-\frac12\right) \log z - z + \frac12\,\log(2\pi) +
\int_0^\infty \frac{\de s}{s} \(\eta-\hf-\frac{1}{s}+\frac{e^{(1-\eta)s}}{e^s-1} \) e^{-z s},
\label{genBinet}
\ee
which follows from \eqref{Binet1} by using the identities \eqref{int-expab}.

\subsubsection*{Barnes function}

The Barnes function (see e.g. \cite{Vigneras1979}) is the unique meromorphic function $G(z)$
which satisfies  the functional equation
\be
G(z+1)=\Gamma(z)\, G(z)\, ,
\label{eqGz}
\ee
the convexity condition $\frac{d^3}{dz^3}\log G(z)\geq 0$ for all $z>1$, and
the normalization condition $G(1)=1$.
It has the following useful properties:
\begin{itemize}
\item
its logarithmic derivative is
\be
\label{eqdG}
\frac{\de}{\de z} \log G(z+1) = \frac12\log (2\pi) -z + \frac12 + z \frac{\de}{\de z} \log \Gamma(z),
\ee
\item
Kinkelin's reflection formula
\bea
\log\frac{G(1-z)}{G(1+z)}
= z \log \left( \frac{\sin (\pi z)}{2\pi z}\right) -  \int_0^z \pi x\, \log\sin(\pi x)\, \de x,
\label{reflectG}
\eea
\item
two integral representations of Binet type, both valid for valid for $\Re z>0$ \cite{Adamchik2003},\footnote{The first formula
corrects a sign misprint in \cite[Eq.(25)]{Adamchik2003}.
The second is obtained from \cite[Eq.(20)]{Adamchik2003} by twice integrating by parts.
Note that the second integration produces the boundary contribution $-\frac{1}{12}\log(z)$.}
\be
\begin{split}
\log G(z+1)=&\,z\log\Gamma(z)+ \frac{z^2}{4} -\frac{B_2(z)}{2}\, \log z +\zeta'(-1) -\frac{1}{12}
\\
&\,-  \int_0^\infty \frac{\de s}{s^2}
\left( \frac12+ \frac{1}{s} +\frac{s}{12}  - \frac{1}{1-e^{-s}} \right) e^{-zs},
\end{split}
\label{BinetG1}
\ee
and
\be
\label{BinetG2}
\begin{split}
\log G(z+1) =&\,   \frac{1}{2}\(z^2-\frac16\) \log z-\frac{3z^2}{4} + \zeta'(-1) + \frac{z}{2}\, \log (2\pi)
\\
&\,
+\frac{1}{2\pi^2}\int_0^{\infty} \frac{s \de s }{s^2+z^2}\,
\Bigl[ 2\pi s  \log\(1-e^{-2\pi s} \) -  \Li_2\(e^{-2\pi s} \)  \Bigr].
\end{split}
\ee
where $B_2(z)$ is the Bernoulli polynomial and $\zeta'(-1)$ is related to Glaisher's constant
$\gamma_E$ via $\zeta'(-1)=\frac{1}{12}-\log\gamma_E$.
\end{itemize}
We will need a generalization of the first integral representation similar to Hermite's generalization \eqref{genBinet}
of the first Binet formula. It is given by the following
\begin{proposition}
For $\Re z>0$ and $\Re (z+\eta)>0$, one has
\be
\begin{split}
\log G(z+\eta+1) =&\,  \zeta'(-1) -\frac34 \,z^2-z\eta +\frac12\, (z+\eta)\log (2\pi)
+\frac12\( (z+\eta)^2-\frac{1}{6}\) \log z
\\
&\, + \int_0^\infty \frac{\de s}{s}
\left( \frac{1}{s^2}-\frac{\eta}{s}+\frac{\eta^2}{2} -\frac1{12}  - \frac{e^{(1-\eta)s}}{(e^s-1 )^2}
\right) e^{-zs}.
\end{split}
\label{BinetG1gen}
\ee
\end{proposition}
\begin{proof}
First, let us note the following identity
\bea
\int_0^\infty \frac{\de s}{s}
\( \frac{1}{s^2} -\frac1{12}  - \frac{e^{s}}{(e^s-1 )^2}\) e^{-zs}
&=& z\log \Gamma(z)-z\(z-\hf\)\log(z)+z^2-\frac{z}{2}\log(2\pi)-\frac{1}{12}
\nn\\
&&
-\int_0^\infty \frac{\de s}{s^2}
\(\frac12+ \frac{1}{s} +\frac{s}{12} - \frac{e^s}{e^s-1} \) e^{-z s},
\eea
which can be established by integration by parts and the use of the first Binet formula \eqref{Binet1}.
Therefore, the integral representation \eqref{BinetG1} can be rewritten as
\be
\begin{split}
\log G(z+1)=&\,\(\frac{z^2}{2}-\frac{1}{12}\) \log z -\frac{3z^2}{4}+\frac{z}{2}\log(2\pi) +\zeta'(-1)
\\
&\,+\int_0^\infty \frac{\de s}{s}
\( \frac{1}{s^2} -\frac1{12}  - \frac{e^{s}}{(e^s-1 )^2}\) e^{-zs}.
\end{split}
\label{BinetG1new}
\ee
Let us now substitute $z\to z+\eta$. Then the last integral term differs from the one in \eqref{BinetG1gen}
by the following contribution
\be
\begin{split}
\int_0^\infty \frac{\de s}{s} & \[
\( \frac{1}{s^2} -\frac1{12}\) e^{-(z+\eta) s}
-\( \frac{1}{s^2}-\frac{\eta}{s}+\frac{\eta^2}{2} -\frac1{12}\) e^{-zs}\]
\\
& \quad = \frac{3\eta^2}{4}+\frac{z\eta}{2}-\hf\((z+\eta)^2-\frac16\)\log\frac{z+\eta}{z}\, ,
\end{split}
\ee
where we used \eqref{int-expab}.
Combining this result with terms in the first line of \eqref{BinetG1new} after the replacement $z\to z+\eta$, one reproduces
the first line in \eqref{BinetG1gen}, which completes the proof.
\end{proof}

\subsection{Generalized Gamma function \label{sec_appL}}

The generalized Gamma function
 $\Lambda(z,\eta)$ is a function on $\IC^\times\times\IC$ defined by \cite{Barbieri:2018swu}
\be
\label{defLambda}
\Lambda(z,\eta) =\frac{e^z\,\Gamma(z+\eta)}{\sqrt{2\pi} z^{z+\eta-1/2}}\, .
\ee
This function has the following properties:
\begin{itemize}
\item
twisted periodicity
\be
\label{Lperiod}
\Lambda(z,\eta+1) = \frac{z+\eta}{z}\, \Lambda(z,\eta),
\ee

\item
reflection property
\be
\Lambda(-z,\eta) \, \Lambda(z,1-\eta) =
\begin{cases}
\left(1-e^{2\pi\I( z-\eta)} \right)^{-1}, \quad & \Im z>0,
\\
\left(1-e^{2\pi\I(\eta-z)} \right)^{-1}, \quad & \Im z<0,
\end{cases}
\label{shiftLam}
\ee

\item
two integral representations directly following from the generalized version of the first Binet formula \eqref{genBinet}
and the second Binet formula \eqref{Binet2}, respectively,
\bea
\log \Lambda(z,\eta)&=&
\int_0^\infty \frac{\de s}{s} \(\eta-\hf-\frac{1}{s}+\frac{e^{(1-\eta)s}}{e^s-1} \) e^{-z s}
\label{Lam-Binet}
\\
&=&\,
\(z+\eta-\hf\)\log\frac{z+\eta}{z}-\eta-\frac{z+\eta}{\pi}\int_0^\infty \de s\,\frac{ \log\(1-e^{-2\pi s}\)}{s^2+(z+\eta)^2},
\label{logL}
\eea
where the first representation is valid for $\Re z>0$, $\Re (z+\eta)>0$,
whereas the second requires only the second condition.

\item
asymptotic series at large $z$ \cite[\S8.3]{Bridgeland:2019fbi}
\be
\label{Lambdaexp}
\log\Lambda(z,\eta)  = \sum_{k= 2}^\infty\frac{(-1)^k B_k(\eta)}{k(k-1)}\, z^{1-k}.
\ee
It follows from \eqref{Lam-Binet} by using
$\frac{s e^{x s}}{e^s-1}=\sum_{k=0}^{\infty} \frac{s^k}{k!} B_k(x)$ and integrating term by term
using $\int_0^\infty s^n e^{-zs}\de s= n!/ z^{n+1}$.
Note that this asymptotic expansion is consistent with the periodicity relation \eqref{Lperiod},
since $B_k(\eta+1)-B_k(\eta)=k \eta^{k-1}$.

\end{itemize}

\subsection{Generalized Barnes function \label{sec_appU}}

We define the generalized Barnes function $\Upsilon(z,\eta)$ on $\IC^\times \times \IC$ by
\be
\label{defUpsilon}
\Upsilon(z,\eta) = \frac{e^{\frac34z^2- \zeta'(-1)}\, G(z+\eta+1)  }
{(2\pi)^{z/2} \,z^{\frac12 z^2} \bigl[\Gamma(z+\eta) \bigr]^\eta}\, ,
\ee
where $G(z)$ is the Barnes function.
This definition generalizes the function
$\Upsilon(z)$ introduced  in \cite[Eq.(17)]{Bridgeland:2016nqw}
which is obtained from \eqref{defUpsilon} setting $\eta=0$.
It has the following properties:
\begin{itemize}
\item
relation to the generalized Gamma function \eqref{defLambda}
\be
\label{dUps}
\frac{\p}{\p z}  \log \Upsilon(z,\eta) = z \frac{\p}{\p z}  \log \Lambda(z,\eta),
\ee
which follows from the property of the Barnes function \eqref{eqdG},
\item
twisted periodicity
\be
\label{Uperiod}
\Upsilon(z,\eta+1) = (z+\eta)^{-\eta} \Upsilon(z,\eta),
\ee

\item
reflection property
\be
\log\frac{\Upsilon(-z,\eta)}{\Upsilon(z,1-\eta)} =  -\frac{\pi\I}{12}\,B_2(\eta)+
\begin{cases}
z\log\left(1-e^{2\pi\I( z-\eta)} \right)+\frac{1}{2\pi\I}\, \Li_2\left( e^{2\pi\I( z-\eta)} \right), & \Im z>0
\\
z\log \left(1-e^{2\pi\I(\eta-z)} \right) -\frac{1}{2\pi\I}\, \Li_2\left(  e^{2\pi\I(\eta-z)}\right),
 & \Im z<0,
\end{cases}
\label{Yreflect}
\ee
which is a consequence of \eqref{reflectG},

\item
two integral representations of Binet type
\bea
\log\Upsilon(z,\eta) &=&
\int_0^\infty \frac{\de s}{s} \( \frac{1}{s^2} - \frac12 \,B_2(\eta)
- \frac{\eta(e^s-1)+1}{(e^s-1)^2}\,e^{(1-\eta)s} \)  e^{-z s}
\label{BinetUps1}\\
&& -\frac{1}{2} \,B_2(\eta) \log z
\nn
\\
&=& \frac{z^2}{2}\, \log \frac{z+\eta}{z}-\hf\,B_2(\eta)\, \log (z+\eta)-\frac{z\eta}{2}+\frac{\eta^2}{4}
\label{BinetUps2}\\
&&
+\int_0^{\infty} \frac{\de s }{s^2+(z+\eta)^2}
\left[ \frac{1}{\pi}\(s^2+\eta(z+\eta)\)  \log\(1-e^{-2\pi s} \) - \frac{s}{2\pi^2}\, \Li_2\(e^{-2\pi s} \)  \right],
\nn
\eea
directly following from \eqref{BinetG1gen}, \eqref{genBinet} and
\eqref{BinetG2}, \eqref{Binet2}, respectively,

\item
asymptotic series at large $z$
\be
\label{Upsexp}
\log\Upsilon(z,\eta) = -\frac{1}{2}\, B_2(\eta) \log z
+ \sum_{k= 3}^\infty  \frac{ (-1)^k B_k(\eta)}{k(k-2)}\, z^{2-k},
\ee
which follows from \eqref{BinetUps1} taking into account
\be
 \frac{1}{s^2} - \frac12\, B_2(\eta) - \frac{\eta(e^s-1)+1}{(e^s-1)^2}\,e^{(1-\eta)s}
 = \sum_{k=3}^{\infty} \frac{(k-1)}{k!}\, B_k(1-\eta)\, s^{k-2}
\ee
and integrating term by term.
For $\eta=0$, this reduces to the result given in \cite[\S 5.4]{Bridgeland:2016nqw}.

\end{itemize}

\section{Solution in the uncoupled finite case}
\label{ap-finite}

Here we provide some intermediate steps in deriving \eqref{Xuncoupled} and \eqref{taufun-res2}.

\subsection{Darboux coordinates}
\label{sec-Xfinite}

In the uncoupled case, $\cX_{\gamma'}$ appearing in the r.h.s. of the integral equation \eqref{TBAeq}
can be replaced by $\cXsf_{\gamma'}$. Then, after combining contributions of opposite charges,
this equation takes the form
\be
\label{Xfinite}
\cX_\gamma(t) = \cXsf_{\gamma}(t) \,\exp\Biggl[ \sum_{\gamma'\in \Gamma \atop \Re(Z_{\gamma'}/t)>0}
\Omega(\gamma') \langle \gamma,\gamma'\rangle  R_{\gamma'}(t)\Biggr],
\ee
where
\be
R_{\gamma}(t)
=\frac{1}{2\pi\I}\int_{-\I\vth_\gamma}^{\infty} \de s \frac{\log\(1-e^{-2\pi s}\)}{s-\I (Z_\gamma /t-\vth_\gamma)}
-\frac{1}{2\pi\I}\int_{\I\vth_\gamma}^{\infty}\frac{\log\(1-e^{-2\pi s}\)}{s+\I (Z_\gamma /t-\vth_\gamma)}
\ee
and we changed the integration variable $t'=\I Z_{\gamma}/(s+\I\vth_\gamma)$.
Splitting each integral into two parts, one obtains
\be
\begin{split}
R_{\gamma}(t)
=&\, \frac{Z_\gamma /t-\vth_\gamma}{\pi}\int_0^{\infty}\de s\, \frac{\log\(1-e^{-2\pi s}\)}{s^2+(Z_\gamma /t-\vth_\gamma)^2}
\\
&\,
-\frac{1}{2\pi\I}\int_0^{-\I\vth_\gamma} \de s\,\frac{\log\(1-e^{-2\pi s}\)}{s-\I(Z_\gamma /t-\vth_\gamma)}
+\frac{1}{2\pi\I}\int_0^{-\I\vth_\gamma} \de s\, \frac{\log\(1-e^{2\pi s}\)}{s-\I(Z_\gamma /t-\vth_\gamma)}\, ,
\end{split}
\label{Reval}
\ee
where in the last term we flipped the sign of $s$. The last two terms can be combined using \eqref{Li-ident} at $n=1$,
leading to the contribution
\be
-\int_0^{-\I\vth_\gamma} \de s\,\frac{[\I s]-\hf}{s-\I(Z_\gamma /t-\vth_\gamma)}\, .
\ee
If $\Re \vth_\gamma\in (0,1)$, the integral is easily evaluated to
\be
-\(Z_\gamma /t-\vth_\gamma+\hf\)\log\(1-\frac{t\vth_\gamma}{Z_\gamma}\)-\vth_\gamma\, .
\ee
Combining this with the first term in \eqref{Reval}, one recognizes the integral representation \eqref{logL}
of the generalized Gamma function shifted by a logarithmic term
\be
-\log \Lambda\(\frac{Z_\gamma}{t}, -\vth_\gamma\)-\log\(1-\frac{t\vth_\gamma}{Z_\gamma}\)
=-\log \Lambda\(\frac{Z_\gamma}{t},1-\vth_\gamma\),
\label{Rincomplet}
\ee
where in the last step we used the twisted periodicity property \eqref{Lperiod}.
If however $\Re\vth_\gamma\notin (0,1)$, there is an additional contribution
given by
\be
\eps_\gamma\sum_{n\in S_{\eps_\gamma}}\int_{-\I n}^{-\I\vth_\gamma}\frac{\de s}{s-\I(Z_\gamma /t-\vth_\gamma)}
=-\eps_\gamma \sum_{n\in S_{\eps_\gamma}}\log\(1-\frac{t(\vth_\gamma-n)}{Z_\gamma}\),
\label{contr-brack}
\ee
where $\eps_\gamma=\sgn(\Re\vth_\gamma)$ and we introduced two finite sets
$S_+=\{1,\dots, \fl{\Re \vth_\gamma}\}$ and $S_-=\{0,\dots, 1+\fl{\Re \vth_\gamma}\}$.
Adding this contribution to \eqref{Rincomplet} and again using repeatedly the twisted periodicity \eqref{Lperiod},
one finally obtains
\be
R(t) =-\log \Lambda\(\frac{Z_\gamma}{t},1-[\vth_\gamma]\).
\ee
Substituting this result into \eqref{Xfinite}, one gets the formula \eqref{Xuncoupled} given in the main text.

\subsection{$\tau$ function}
\label{sec_Apptau}

The starting point for deriving the $\tau$ function is the formula \eqref{restau}.
In terms of the integer BPS indices (and after discarding an irrelevant function of $\theta^a$
coming from $\log(\pm n)$ in the last term where $n$ counts the reducibility of the charge),
it takes the form
\bea
\log\tau &=& \frac{1}{4\pi^2}\sum_{\gamma\in\Gamma_{\!\star}}
\biggl[\Omega(\gamma)\int_{\ell_\gamma} \frac{\de t'}{t'}\, \frac{t}{t'-t}\,
\Bigl( \Li_2\(e^{2\pi\I(\vth_\gamma- Z_\gamma/t')}\)
-2\pi\I\, \frac{Z_\gamma}{t'}\,\log\(1-e^{2\pi\I(\vth_\gamma- Z_\gamma/t')}\)\Bigl)
\nn\\
&&
-\Li_2\(e^{2\pi\I\vth_\gamma}\)\log (Z_\gamma/t)\biggr].
\label{logtau}
\eea
The next steps are similar to the one in the previous appendix.
Setting $t'=\I Z_{\gamma}/(s+\I\vth_\gamma)$ and combining the contributions of opposite charges, one finds
\bea
\log\tau
&=& -\frac{1}{4\pi^2}\sum_{\gamma\in\Gamma_{\!\star}}
\sigma_\gamma\Omega(\gamma)\biggl[ \int_{-\I\vth_\gamma}^\infty \de s\,
\frac{\Li_2\(e^{-2\pi s}\) -2\pi(s+\I\vth_\gamma)\,\log\(1-e^{-2\pi s}\)}
{s-\I(Z_\gamma /t-\vth_\gamma)}
\nn\\
&&\qquad
+\Li_2\(e^{2\pi\I\vth_\gamma}\)\log (Z_\gamma/t)\biggr]
\nn\\
&=&\sum_{\gamma\in \Gamma \atop \Re(Z_{\gamma}/t)>0}
\sigma_\gamma\Omega(\gamma)\[
\int_{0}^\infty \de s \frac{-\frac{s}{2\pi^2 }\, \Li_2\(e^{-2\pi s}\)
+\frac{1}{\pi}\(s^2+\vth_\gamma(\vth_\gamma-Z_\gamma/t)\)\log\(1-e^{-2\pi s}\)}
{s^2+(Z_\gamma/t-\vth_\gamma)^2}
\right.
\nn\\
&&\left.
+\hf\int_0^{-\I\vth_\gamma} \de s\,
\frac{B_2([\I s]) -(\I s-\vth_\gamma)(2[\I s]-1)}
{s-\I(Z_\gamma /t-\vth_\gamma)}
-\hf\, B_2([\vth_\gamma])\,\log (Z_\gamma/t)\].
\label{compute-tau}
\eea
For $\Re \vth_\gamma\in (0,1)$, the second integral is evaluated to
\be
\frac14\, \vth_\gamma^2 +\frac{\vth_\gamma Z_\gamma}{2t }
+\hf\(\frac{Z_\gamma^2}{t^2}-B_2(\vth_\gamma)\)\log\(1-\frac{t\vth_\gamma}{Z_\gamma}\).
\ee
Substituting this back into \eqref{compute-tau}, one recognizes the square bracket as
the integral representation \eqref{BinetUps2} of the generalized Barnes function $\Upsilon(z,\eta)$
shifted by a logarithmic term
\be
\log \Upsilon\(\frac{Z_\gamma}{t}, -\vth_\gamma\)+\vth_\gamma \log\(Z_\gamma/t-\vth_\gamma\)
=\log \Upsilon\(\frac{Z_\gamma}{t},1-\vth_\gamma\),
\label{Tau-incomplet}
\ee
where in the last step we used the twisted periodicity property \eqref{Uperiod}.
For $\Re\vth_\gamma\notin (0,1)$, there are additional contributions which can be computed as in \eqref{contr-brack}
and amount to replacing $\vth_\gamma$ by $[\vth_\gamma]$, consistently with the periodicity of the original integral representation.
As a result, one arrives at the formula \eqref{taufun-res2}.

\providecommand{\href}[2]{#2}\begingroup\raggedright\endgroup


\end{document}